%% file: main_journal.tex
\documentclass[aps,prx,superscriptaddress,10pt]{revtex4-2} 

\usepackage{amsmath, amssymb, graphicx}
\usepackage{xcolor}
\usepackage{braket}
\usepackage{physics}
\usepackage{algorithm}
\usepackage{algpseudocode}
\usepackage{amsthm}
\usepackage{yhmath}
\usepackage{mathtools}
\usepackage{nicefrac}
\usepackage{faktor}
\usepackage{tensor}
\usepackage{hyperref} 
\usepackage{tikz}
\usetikzlibrary{calc, decorations.pathreplacing,arrows.meta,positioning}

\newcommand{\lset}{\mathcal{L}}
\newcommand{\qset}{\mathcal{Q}}

\newcommand{\hil}{\mathcal{H}}
\newcommand{\vect}{\operatorname{Vect}}

\newcommand{\E}{\mathbb{E}}
\newcommand{\conv}{\mathrm{Conv}}
\newcommand{\ext}{\mathrm{Ext}}
\newcommand{\pqp}{\mathrm{PQP}}
\newcommand{\qp}{\qset_p}
\renewcommand{\pqp}{\qp}

\newcommand{\Z}{\mathbb{Z}}
\newcommand{\R}{\mathbb{R}}

\newcommand{\littletaller}{\mathchoice{\vphantom{\big|}}{}{}{}}
\newcommand{\restr}[2]{{
  \left.\kern-\nulldelimiterspace 
  #1 
  \littletaller 
  \right|_{#2} 
  }}
\newcommand{\pmap}{\mathcal{P}}
\newcommand{\bop}{\mathcal{B}}

\newtheorem{theorem}{Theorem}
\newtheorem{definition}{Definition}
\newtheorem{lemma}{Lemma}
\newtheorem{proposition}{Proposition}
\newtheorem{corollary}{Corollary}

\begin{document}

\title{The perturbative method for quantum correlations}
\author{Sacha Cerf, Harold Ollivier}
\affiliation{QAT team, DIENS, École Normale Supérieure, PSL University, CNRS, INRIA, 45 rue d’Ulm, Paris 75005, France}

\begin{abstract}
  The set $\qset$ of quantum correlations is the collection of all possible probability distributions on measurement outcomes achievable by space-like separated parties sharing a quantum state. Since the original work of Tsirelson \cite{tsirelson1980quantum}, this set has mainly been studied through the means algebraic and convex geometry techniques.

  We introduce a perturbative method using Lie-theoretic tools for the unitary group to analyze the response of the evaluations of Bell functionals under infinitesimal unitary perturbations of quantum strategies. Our main result shows that, near classical deterministic points, an $(n, 2, d)$ Bell operator decomposes into a direct sum of $(k, 2, d-1)$ Bell operators which we call \emph{subset games}. We then derive three key insights: (1)~in the $(n, 2, 2)$ case, if $p_0$ is classically optimal, it remains locally optimal even among 2-dimensional quantum strategies, implying in turn that the boundary of $\qset$ is flat around classical deterministic points; (2)~it suggests a proof strategy for Gisin's open problem on correlations in $\mathcal{Q}(D)$ unattainable by projective strategies of the same dimension; and (3)~it establishes that the Ansatz dimension is a critical resource for learning in distributed scenarios, even when the optimal solution admits a low-dimensional representation.

\end{abstract}

\maketitle
\input{sections_journal_v3/introduction.tex}
\input{sections_journal_v3/preliminaries.tex}

\input{sections_journal_v3/section3.tex}

\input{sections_journal_v3/section4.tex}
\input{sections_journal_v3/section5.tex}
\input{sections_journal_v3/conclusion.tex}
\bibliographystyle{apsrev4-2}
\bibliography{ref}

\newpage
\appendix
\input{sections_journal_v3/appendix.tex}
\end{document}

%% file: sections_journal_v3/introduction.tex
\section{Introduction}

\paragraph{Quantum correlations and nonlocality}
The set of quantum correlations $\qset$ lies at the core of the study of quantum nonlocality. It is defined as the collection of measurement statistics---represented as high-dimensional vectors---achievable by space-like separated parties sharing a quantum state. First formalized by Tsirelson~\cite{tsirelson1980quantum}, it is strictly larger than its classical counterpart \(\mathcal{L}\) of statistics explainable by local hidden variable models~\cite{ bell1964einstein}, demonstrating a quantum advantage in the correlations achievable by entangled parties. The comparative study of $\qset$ and $\lset$ is mainly conducted through the means of violations of \emph{Bell inequalities}, which are linear inequalities satisfied by every point in $\lset$. The strict inclusion $\lset \subsetneq \qset$ is a foundational result which enabled an experimental proof of the existence of entanglement and the nonlocal nature of quantum mechanics~\cite{aspect1982experimental}. Although a full characterization of $\qset$ has been proven to be impossible due to the complexity-theoretic result $\textrm{MIP}^*=\textrm{RE}$~\cite{ji2022mipre}, its understanding remains a central challenge in quantum foundations.

\paragraph{Nonlocal games}
A nonlocal game is an interactive protocol between a referee and a group of players. The referee sends each player a question, the players respond without communicating between them, and the referee decides whether they win based on the received answers. The game is called nonlocal if the players can achieve a higher winning probability than what classical physics allows. They do it using an entangled quantum state, and deciding their answer using a measurement depending on the question. The zoology of nonlocal games provides a rich framework to explore the quantum-classical divide in correlations. For instance:
\begin{itemize}
\item[--]The CHSH game~\cite{clauser1969proposed} is the simplest Bell non-locality scenario, involving two parties, two measurement settings, and two outcomes. It demonstrates that quantum mechanics allows for correlations that violate classical bounds, with a maximal violation known as the Tsirelson bound.
\item[--] Multipartite games, such as the GHZ game~\cite{greenberger1990bell}, involve three or more parties and reveal stronger forms of non-locality. These games are essential for understanding the scalability of quantum advantage in multi-party settings.
\item[--] Graph-based games extend the framework to more complex interaction structures, where parties are arranged according to a graph, and correlations are constrained by the graph's topology. These games are particularly relevant for studying quantum networks and distributed quantum computing~\cite{renou2019limits,tavakoli2022bell}.
\end{itemize}

\paragraph{Previous work} 
Since Tsirelson’s seminal work~\cite{tsirelson1980quantum}, the study of nonlocality has primarily relied on convex geometric and algebraic methods~\cite{masanes2005extremalquantumcorrelationsn, Le2023quantumcorrelations, Goh_2018}, as well as hierarchies of semidefinite programs derived from operator algebras~\cite{navascues2008convergent, xu2023quantum, renou2022bilocal, ligthart2023inflation}. Indeed, optimization of nonlocal strategies has been succesfully using SDP hierarchies to accomplish this task~\cite{navascues2008convergent,ligthart2023inflation, tavakoli2024semidefinite}.

A new line of research on this topic is emerging by importing variational methods developed in the context of hybrid quantum-classical optimization~\cite{mcclean2016theory, cerfontaine2021normalized} and quantum machine learning~\cite{stokes2020quantum, sweke2020stochastic}. The idea is to use variational methods to learn optimal nonlocal strategies. Early results~\cite{furches2025applicationlevelbenchmarkingquantumcomputers, kerenidis2025quantumagentsalgorithmicdiscovery} are encouraging and point toward the relevance of learning frameworks to probe the structure of \(\qset\), especially around local deterministic points. Indeed, learning puts an emphasis on the need to better understand the geometric structure of the set of quantum correlations $\qset$. This is because the ability to learn is tied to the dynamical behavior of quantum strategies under perturbations. Yet, this topic has received very little attention until now. To our knowledge, this question has been raised only in~\cite{Le2023quantumcorrelations}, which introduces a trigonometric parametrization of the extremal points that enables differentiation, but while being limited to the bipartite scenario.

\paragraph{Proposed approach and results}
In this paper, we propose a perturbative method using Lie-theoretic tools on the unitary group. We analyze the response of the evaluation of a Bell functional under infinitesimal unitary perturbations of a quantum strategy. This dynamical perspective allows us to simplify the optimization problem around specific points of interest, namely the local deterministic correlations which form the vertices of the local polytope $\mathcal{L}$.

Our approach yields three primary contributions:
\begin{enumerate}
\item A dimension reduction theorem showing that the second-order variation of the evaluation of a Bell functional around a deterministic point in the $(n, 2, d)$ scenario is governed by effective \emph{subset games} with reduced output dimension ($d-1$).
\item A proof that in the $(n, 2, 2)$ scenario, the boundary of $\mathcal{Q}$ is flat around every local deterministic point, meaning---in a specific way stated later---that there is no quantum extremal point in their immediate neighborhood. To our knowledge, this is the first mathematical statement on the geometry of the multipartite quantum set.
\item A proposed pathway to resolve an open problem formulated by Gisin on the existence, for a maximum allowed local dimension, of quantum correlations requiring non-projective POVMs that maximally violate some Bell inequality~\cite{gisin2007bell} (see question 10). This question is fundamental for quantum theory: it asks if POVMs, objects that naturally arise in the mathematical axiomatic of measurements, generalizing PVMs, are in a sense \emph{relevant} for quantum nonlocality. It has been answered positively under the restriction that the parties share a Bell state~\cite{Vertesi_2010}. Exhibiting a $(n, 2, d)$ Bell operator for which all the subset game operators at the classical optimum $p_0$ are negative, but $p_0$ is not a local optimum in $\qset(d)$---which could be achieved using the perturbative machinery we propose---would prove the existence of such correlations, with no constraint on the shared state.
\end{enumerate}

Our findings also position again nonlocal games as a compelling toy model for studying quantum variational learning (see~\cite{furches2025applicationlevelbenchmarkingquantumcomputers}), revealing fundamental limitations even in seemingly simple scenarios. Specifically, in the \((n, 2, 2)\) setting---where qubit strategies are known to suffice for maximizing any Bell functional---we demonstrate that efficient learning via variational quantum models using qubits is hard, particularly when initialized from the optimal classical correlation. Thus, the dimension of the quantum Ansatz itself is a critical resource for learning, independent of the minimal dimension required to represent the optimal solution.

%% file: sections_journal_v3/preliminaries.tex
\section{Preliminaries}
Following the usual convention of~\cite{donohue2015}, we denote the Bell scenario of $n$ space-like separated parties with $m$ different measurement settings each having $d$ possible outcomes by $(n, m, d)$. In this work, we focus on the $(n, 2, d)$ scenario, for arbitrary $n$, and index the setting and outcome with $\Z_2$ and $\Z_d$, respectively.

\paragraph{Correlations}
The setting--outcome statistics of the $n$ parties can be compiled in a real vector:
\begin{equation} p = (p(a \vert x))_{a \in \Z_d^n, x \in \Z_2^n} \in \R^{d^n \times 2^n},
\end{equation} where $a$ and $x$ are the strings of size $n$ denoting respectively the outcome and the setting of each of the $n$ players. We refer to such a real vector as a \emph{correlation}, or \emph{correlation point}.

\paragraph{Bell functional}
A Bell functional \cite{bell1964einstein} $\vec \beta$ is a linear form on correlations. These functionals typically emerge as the expected gain in so called \emph{linear games}, where each of the $n$ parties receive a bit of the input string $x$ from a referee, sampled according to a probability distribution $\pi$. Each party then produces a bit of the output bitstring $a$ following a strategy realizing an $n$-partite correlation $p$, and sends it back to the referee who computes a reward $g(a, x) \in \R$. The expected gain can be re-written with the help of the associated Bell functional as:
\begin{equation} \vec \beta \cdot p = \E_{x \sim \pi}\left[\E_{a \sim p(\cdot \vert x)}\left[g(a, x)\right]\right] = \sum_{ a \in \Z_d^n, x \in \Z_2^n} \pi(x)p(a \vert x)g(a, x).
\end{equation} The coordinates of the vector representing $\beta$ in the canonical basis of $\R^{d^n \times 2^n}$ are $\beta_{a, x}= \pi(x)g(a, x)$, which makes it clear that every Bell functional can be represented as the expected gain of a linear game. Optimizing a strategy for a game then amounts to maximizing a Bell functional over the set of correlations attainable with the given resources.

\paragraph{Local correlations, local strategies}
A correlation is \emph{local} if it is explainable by classical information shared between the $n$ parties (the hidden variable) prior to the choice of measurement setting~\cite{bell1964einstein, einstein35}. Formally, we say that $p \in \R^{d^n \times 2^n}$ is a local correlation if there exists a (non-necessarily unique) \emph{local strategy}, denoted $\Sigma_{\lset} = (\Lambda, \mu, (p^{(i)}_\lambda(\cdot \vert x_i))_{1 \le i \le n, 1 \le \lambda \le \Lambda, x_i \in \Z_2})$, such that:
\begin{equation} p(a \vert x) = \pmap(\Sigma_\lset)(a \vert x) \coloneq \sum_{\lambda = 1}^\Lambda \mu(\lambda)p_\lambda^{(1)}(a_1 \vert x_1) \ldots p_\lambda^{(n)}(a_n \vert x_n),
\end{equation} where $\Lambda \in \mathbb{N}$, $\mu$ is a probability distribution on $\{1, \ldots, \Lambda\}$, for all $1 \le i \le n, 1 \le \lambda \le \Lambda, x_i \in \Z_2$, $p_\lambda^{(i)}(\cdot \vert x_i)$ is a probability distribution on $\Z_d$, and subscripts on strings denote the index.

The correlation $p$ is \emph{local deterministic} if $\Lambda = 1$, and $p^{(i)}(a_i \vert x_i) = \delta_{a_i \tilde{a}_i(x_i)}$ for some $\tilde a_i(x_i) \in \Z_d$. Hence, a local deterministic correlation is uniquely characterized by a family of $n$ gates $\tilde a = (\tilde a_i : \Z_2 \rightarrow \Z_d)_{1 \le i \le n}$, which we call a \emph{local deterministic strategy}. We denote $\pmap(\tilde a)$ the local deterministic correlation associated to $\tilde a$.

We denote by $\lset \subseteq \mathbb{R}^{d^n \times 2^n}$ the set of local correlations. It is well established that $\lset$ is a convex set, the extremal points of which are the local deterministic correlations, and is thus a polytope~\cite{fine1982hidden}. The faces are the locus of maximization of the Bell functionals. The maximization of a Bell functional defines a linear inequality satisfied by every point in $\lset$, called a Bell inequality~\cite{clauser1969proposed}.

\paragraph{Quantum correlations, quantum strategies}
Quantum correlations are the correlations that the $n$ parties can achieve by performing local measurements on a shared quantum state. Formally, we say that $p \in \R^{d^n \times 2^n}$ is a quantum correlation if there exists a \emph{quantum strategy} $\Sigma_\qset = ((\hil_i)_{1 \le i \le n}, \rho, (M^{(i)}_{x_i})_{1 \le i \le n, x_i \in \Z_2})$, where, for $1 \le i \le n, x_i \in \Z_2$, $\hil_i$ is a Hilbert space, $M_{x_i}^{(i)} = \{M^{(i)}_{0 \vert x_i}, M^{(i)}_{d - 1\vert x_i}\}$ is a POVM, and $\rho$ is a quantum state on $\hil_1 \otimes \ldots \otimes \hil_n$, such that:
\begin{equation} p(a \vert x) = \pmap(\Sigma_\qset) \coloneq \Tr\left(\rho M^{(1)}_{a_1 \vert x_1} \otimes \ldots \otimes M^{(n)}_{a_n \vert x_n}\right)
\end{equation} We denote by $\qset$ the set of quantum correlations, which is convex~\cite{tsirelson1980quantum}.

\paragraph{Extremal points, exposed points, and border}
For an arbitrary convex part $C$ of $\R^N$, denote $\ext(C)$ the set of \emph{extremal points} of $C$, which are the points that cannot be expressed as a convex combination of different points in $C$. Extremal points are contained---in general, strictly---in the \emph{border} $\partial C$ of $C$, consisting of points no neighborhood of which lies fully in $C$. An extremal point is \emph{exposed} if it is a unique maximum in $C$ for some linear functional. Extremal points (or vertices) of a polytope are always exposed.

%% file: sections_journal_v3/section3.tex
\section{Response of bell functionals under unitary perturbations}
We present a method for studying how the value of a Bell functional $\vec \beta$ evolves under unitary perturbations of a generic strategy. We propose to define, for each possible input bitstring $x \in \Z_2^n$, the \emph{conditional Bell operator} $\bop_x$, which, given a fixed set of POVMs, 
associates any state $\rho$ to the expected value $\Tr(\rho \bop_x)$ of the strategy conditioned on input $x$. Each of these conditional Bell operators can be evolved by perturbating 
the local measurements and also by evolving the state $\rho$, in the Heisenberg picture.

\begin{definition}
    Let $M = (M^{(i)}_{x_i})_{1 \le i \le n, x_i \in \Z_2}$ (i.e. a tuple of arbitrary local POVMs), and $\vec \beta$ a Bell functional.
    For $x \in \Z_2^n$, define a conditional Bell operator as: 
    \begin{equation}
        \bop_x(M) = \sum_{a \in \Z_d^n} \beta_{a, x} M^{(1)}_{a_1 \vert x_1} \otimes \ldots \otimes M^{(n)}_{a_n \vert x_n}.
    \end{equation}
\end{definition}

The Bell operator $\bop(M)$, giving the expected score $\vec \beta \cdot \pmap(M, \rho) = \Tr(\bop(M) \rho)$ for the shared state $\rho$ is then: 
\begin{equation}
    \bop(M)= \sum_{x \in \Z_2^n} \bop_x(M).
\end{equation}

We now turn to describing an infinestimal unitary evolution of a generic strategy $\Sigma = (\rho, M)$.
The \emph{unitary orbit} of $\Sigma$ is the set of strategies obtained by unitarily evolving the state $\rho$ and the 
conditional POVMs for each party. More formally, the unitary orbit of $\Sigma$ is the set
$\{\vec U \cdot \Sigma = (\vec U \cdot \rho, \vec U \cdot M) \coloneq (U_\rho \rho U_\rho^\dagger, 
(U_{i, x_i}M^{(i)}_{x_i} U_{i, x_i}^\dagger))_{1 \le i \le n, x_i \in \Z_2}\}$ where $\vec U$ describes all tuples of 
unitaries acting on the state and local unitaries acting on the POVMs.

We argue that one can keep the POVMs $M^{(i)}_0$ fixed for our purpose. Indeed, any change on the Bell functional due to its evolution can be compensated by a corresponding evolution of $U_\rho$: 
\begin{align}
    \Tr(U_\rho \rho U_\rho^\dagger \bop(\vec U \cdot M)) &= \Tr(\rho U_\rho^\dagger \sum_{a, x} \beta_{a, x} 
    \bigotimes_{i=1}^n U_{i, x_i}M^{(i)}_{a_i \vert x_i}U_{i, x_i}^\dagger U_\rho )\\
    &= \Tr(\rho \left(\bigotimes_{i = 1}^n U_{i, 0}^\dagger U_\rho\right)^\dagger  \sum_{a, x} \beta_{a, x} \bigotimes_{i=1}^n (U_{i, 0}^\dagger U_{i, x_i}) M^{(i)}_{a_i \vert x_i}
    (U_{i, 0}^\dagger U_{i, x_i})^\dagger \bigotimes_{i = 1}^n U_{i, 0}^\dagger U_\rho)
\end{align}

Hence, one can consider perturbations acting only on the second measurement for each party. The tuple $\vec U = (U_\rho, U_1, \ldots, U_n)$ compiles the local unitaries $U_i$ to 
perturb the POVMs $M^{(i)}_{\cdot \vert 1}$, and $U_\rho$ to perturb the state $\rho$. We now present concise 
formulas for the first and second order response of the Bell functional under unitary evolution of $\Sigma$. 
This method can in principle be generalized to any order. 

In all what follows, instead of Hamiltonians $H$, we use skew-Hermitian generators $K = iH$ to simplify notations, which is completely equivalent \cite{humphreys1972introduction}.
\begin{proposition}
    \label{prop:seriesexpgen}
    Let $\vec U \equiv (e^{-K_\rho}, e^{K_1}, \ldots, e^{K_n})$, where $K_i, 1 \le i \le n$ is a local skew-Hermitian operator, and $K_\rho$ is a 
    skew-Hermitian operator on the shared state. Then, the value of the Bell functional on $\Sigma$ admits the following second
    order expansion.
    \begin{equation}
        \label{eq:second_order_generic}
        \resizebox{0.95\linewidth}{!}{$
        \displaystyle \vec \beta \cdot \pmap(\vec U \cdot \Sigma) = \vec \beta \cdot \pmap(\Sigma) + \Tr(\rho \sum_{x \in \Z_2^n} [K_\rho + x \cdot \vec K, \bop_x(M)] + \frac{1}{2}\left([K_\rho + x \cdot \vec K, \bop_x(M)]_2 + [[x \cdot \vec K, K_\rho], \bop_x(M)]\right)) + o(\vert \vert \vec K \vert \vert^2),
        $}
    \end{equation} 
    where we denoted $x \cdot \vec K \equiv \sum_{i=1}^n x_i K^{(i)}_i$, $\vert \vert \vec K \vert \vert^2 \equiv -\sum_{i = 1}^n \Tr(K_i^2) - \Tr(K_\rho^2)$, and $[A, B]_2 \equiv [A, [A, B]]$ the two-fold commutator.

\end{proposition}
For a local operator $A$, the superscript $(i)$ means $A$ acts on site $i$, i.e. $A^{(i)} = I \otimes \ldots \otimes A \otimes \ldots \otimes I$. However, since it is 
always clear that $K_i$ acts on site $i$, we sometimes slightly abuse notations and drop the superscript.
\begin{proof}
    We evolve the conditional Bell operators in the Heisenberg picture: 
    \begin{align}
        \label{eq:heisenbergevolution}
        \vec \beta \cdot \pmap(\vec U \cdot \Sigma) &= \Tr(e^{-K_\rho} \rho e^{K_\rho} \sum_{x \in \Z_2^n} \bop_x(\vec U \cdot M)) \\
        &= \Tr(\rho  \sum_{x \in \Z_2^n} e^{-K_\rho} \bop_x(\vec U \cdot M) e^{-K_\rho})
    \end{align}
    Let us explicit the evolution of $\bop_x$ for a fixed bitstring $x$: 
    \begin{align}
        e^{K_\rho}&\bop_x(\vec U \cdot M) e^{-K_\rho} = e^{K_\rho} \sum_{a \in \Z_d^n} \beta_{a, x} \bigotimes_{i = 1}^n e^{x_iK_i} M^{(i)}_{a_i \vert x_i} e^{-x_iK_i} e^{-K_\rho} \\
        &= e^{K_\rho} e^{x_1K^{(1)}_1 + \cdots + x_nK^{(n)}_n} \left(\sum_{a \in \Z_d^n} \beta_{a, x} \bigotimes_{i = 1}^n M^{(i)}_{a_i \vert x_i}\right) e^{-x_1K^{(1)}_1 - \cdots - x_nK^{(n)}_n} e^{-K_\rho} \\
        &= e^{K_\rho} e^{x \cdot \vec K} \bop_x(M) e^{-x\cdot \vec K} e^{-K_\rho} \\
        &= e^{K_\rho + x \cdot \vec K + \frac{1}{2} [K_\rho, x \cdot \vec K] + o(\vert \vert \vec K \vert \vert^2)} \bop(M) e^{-(K_\rho + x \cdot \vec K + \frac{1}{2} [K_\rho, x \cdot \vec K]) + o(\vert \vert \vec K \vert \vert^2)} \\
        &= \bop_x(M) + [K_\rho + x \cdot \vec K + \frac{1}{2} [K_\rho, x \cdot \vec K], \bop_x(M)] + \frac{1}{2}[K_\rho + x \cdot \vec K + \frac{1}{2} [K_\rho, x \cdot \vec K], \bop_x(M)]_2  + o(\vert \vert \vec K \vert \vert^2) \\
        &= \bop_x(M) + [K_\rho + x \cdot \vec K, \bop_x(M)] + \frac{1}{2}\left([K_\rho + x \cdot \vec K, \bop_x(M)]_2 + [[x \cdot \vec K, K_\rho], \bop_x(M)]\right) + o(\vert \vert \vec K \vert \vert^2),
    \end{align}
    where we used the BCH formula in the fourth line, and the exponential commutator expansion \cite{hall2015} in the fifth line.
    The desired result follows by summing the contributions of each conditional Bell operator in (\ref{eq:heisenbergevolution}).
\end{proof}

As a corollary, note that we obtain the following necessary condition for optimality of a strategy, similar to that obtained in \cite{araújo2026firstorderoptimalityconditionsnoncommutative}.
\begin{corollary}[First order conditions of optimality]
    If $\Sigma$ is an optimal quantum strategy for $\vec \beta$, then: 
    \begin{align}
        [\bop(M), \rho] &= 0\\
        \forall 1 \le i \le n, \Tr_{j\neq i}\left(\left[\sum_{x : x_i = 0} \bop_x(M), \rho \right]\right) &=\Tr_{j\neq i}\left(\left[\sum_{x : x_i = 1} \bop_x(M), \rho\right]\right) = 0
    \end{align}
\end{corollary}

\begin{proof}
    For $\Sigma$ to be optimal, the first order term $\Tr(\rho\sum_{x \in \Z_2^n} [K_\rho + x \cdot \vec K, \bop_x(M)])$ must vanish for all choice of $K_\rho$ and $K_i, 1 \le i \le n$.
    Moreover, one can write: 
    \begin{align}
        \Tr(\rho\sum_{x \in \Z_2^n} [K_\rho + x \cdot \vec K, \bop_x(M)]) &= \sum_{x \in \Z_2^n} \Tr(\rho [K_\rho + x \cdot \vec K, \bop_x(M)]) \\
        &= \sum_{x \in \Z_2^n} \Tr(\rho (K_\rho + x \cdot \vec K) \bop_x(M)) - \Tr(\rho \bop_x(M)(K_\rho + x \cdot \vec K))\\
        &= \sum_{x \in \Z_2^n} \Tr(\bop_x(M)\rho (K_\rho + x \cdot \vec K)) - \Tr(\rho \bop_x(M)(K_\rho + x \cdot \vec K))\\
        &= \sum_{x \in \Z_2^n} \Tr([\bop_x(M), \rho](K_\rho + x \cdot \vec K)).
    \end{align}
    Hence, if one sets $K_i = 0$ for $1 \le i \le n$, the first order term is: 
    \begin{equation}
        \sum_{x \in \Z_2^n} \Tr([\bop_x(M), \rho]K_\rho) = \Tr([\bop(M), \rho]K_\rho) = 0.
    \end{equation}
    We deduce that $[\bop(M), \rho]$ is orthogonal to every skew-Hermitian matrix, which implies it is orthogonal to any operator, and necessarily $[\bop(M), \rho] = 0$. This is expected: for $\rho, M$ to be 
    optimal for $\Tr(\bop(M)\rho)$, $\rho$ must be a convex combination of eigenvectors of $\bop(M)$.

    Now, if one sets $K_\rho = 0$, and $K_j = 0$ for all $j \neq i$, $i$ fixed, the first order term becomes: 
    \begin{equation}
        \sum_{x \in \Z_2^n} \Tr([\bop_x(M), \rho]x_iK_i) = \Tr(\left[\sum_{x : x_i = 1} \bop_x(M), \rho\right]K_i) = 0.
    \end{equation}
    We deduce that $[\sum_{x : x_i = 1} \bop_x(M), \rho]$ is orthogonal to every operator acting on the $i$-th subsystem,
    or equivalently $\Tr_{j\neq i}\left(\left[\sum_{x : x_i = 1} \bop_x(M), \rho\right]\right) = 0$. To obtain the same equation with the sum 
    on $x$ such that $x_i = 0$, we simply insert $[\bop(M), \rho] = 0$ in the equality: 
    \begin{equation}
        \Tr_{j\neq i}\left(\left[\sum_{x : x_i = 0} \bop_x(M), \rho\right]\right) = \Tr_{j\neq i}\left( [\bop(M), \rho] \right) - \Tr_{j\neq i}\left(\left[\sum_{x : x_i = 1} \bop_x(M), \rho \right] \right) = 0.
    \end{equation}

    \textbf{Note :} While unitary perturbations are not sufficient to reconstruct the full neighborhood of a quantum strategy. note that this framework can still allow to recover generic POVM perturbations. One can reparametrize the 
    correlations of local dimension $D$ by the Naimark dilation \cite{naimark1970normed} $\Pi_{ \cdot \vert x}$ of the measurements $M_{\cdot \vert x}$ used to realize them. If $\Sigma = (\rho, M)$ is a quantum strategy, the dilated strategy 
    $D(\Sigma) = (\rho \otimes \ketbra{0}{0}^{\otimes n}, \Pi)$ (where each party holds an extra auxiliary qu-$D$-it in state $\ket{0}$), is then a $D^2$-dimensional projective strategy such that $\pmap(\Sigma) = \pmap(D(\Sigma))$. Unitary perturbations of $D(\Sigma)$ 
    not acting on the auxiliary space should then correspond to the whole neighborhood of $\Sigma$. However, we do not make use of this idea in this paper.
\end{proof}

%% file: sections_journal_v3/section4.tex
\section{Unitary perturbations of local deterministic strategies}
Given a deterministic correlation $p_{\tilde a} = \pmap(( \tilde a_j)_{1 \le j \le n})$, how do 
quantum correlations in its neighborhood compare with respect to a functional $\vec \beta$? This question can be 
adressed by considering all the quantum strategies realizing a specific local deterministic correlation, and using the tools 
from the previous section to obtain the behaviour of $\vec \beta$ under perturbations of these strategies. As we show 
in this section, the response of $\vec \beta$ under unitary perturbations takes a much simpler form when the base point is a local deterministic point.

The interest of this problem is threefold. First, especially in the $(n,2,2)$ scenario, as we will see in the next section, it informs us on the geometry of $\qset$ around local deterministic
points. Second, it investigates the limitations of optimizing quantum strategies via gradient-based methods. Finally, comparing
the behaviour of $\vec \beta$ under perturbations of projective strategies, and those of non-projective POVM strategies of same local dimension
$D$ might reveal the existence of correlation points which can be obtained with the latter type, but not with the former type of strategy. hether such points exist for some strategy dimension $D$ and a number of outputs $d$ is an open problem 
formulated by Gisin in \cite{gisin2007bell}. For example
one could find a functional for which a correlation point $p$ is locally optimal among correlations realized by projective strategies, but not among those
realizable with POVMs of local dimension $D$. Let us define formally what we mean by projective strategies.

\begin{definition}[QP correlations and strategies]
	We say that a quantum strategy $\Sigma_\qset = ((\hil_i)_{1 \le i \le n}, \rho, (M^{(i)}_{x_i})_{1 \le i \le n, x_i \in \Z_2})$ is a qu-$D$-it projection (QP) strategy if:
	\begin{itemize}
		\item for all $1 \le i \le n$, $\hil_i = \mathbb{C}^D$, 
		 \item for all $1 \le i \le n$, $x_i \in \Z_2$, $M^{(i)}_{x_i}$ is a projective valued measurement, characterized by projector elements $(\Pi^{(i)}_{a_i \vert x_i})_{a_i \in \Z_d}$ with $\sum_{a_i = 0}^{d-1}\Pi_{a_i \vert x_i} = I_D$ and $\Pi_{a_i \vert x_i} \Pi_{a'_i \vert x_i} = 0$ for $a_i \neq a'_i$.
	\end{itemize}
	The set of all QP correlations one can realize using a QP strategy with qu-$D$-its in the $n2d$ scenario is denoted $\pqp(D)$, and is a subset of $\qset$.
\end{definition}

If $\qset(D)$ denotes the set of correlations realizable using a quantum strategy 
of local dimension $D$ at most, the problem then restates as : do we have $\qp(D) \subsetneq \qset(D)$ for some $D$? We do not answer this precise question, but we give new elements 
pointing to a positive answer, and a potential strategy to address this conjecture using the tools described in the previous section.

Let us first characterize all the strategies realizing a local deterministic point $p_{\tilde a}$.

\begin{lemma}
    \label{lem:localdeterministic}
    Let $p_{\tilde a}$ be a local deterministic correlation point, and $\Sigma = (\rho, M)$ a quantum strategy. 
    Then $\pmap(\Sigma) = p_{\tilde a}$ if and only if $\rho = \sum_{j = 1}^k p_j \ketbra{\psi_j}{\psi_j}$, such that for all $1 \le j \le k$ 
    $M_{\tilde a_1(x_1) \ldots \tilde a_n(x_n) \vert x}\ket{\psi_j} = \ket{\psi_j}$, and consequently, for $a \in \Z_d^n$, 
    $ M_{a \vert x} \ket{\psi_j} = 0$ if $a \neq \tilde a_1(x_1) \ldots \tilde a_n(x_n)$,
    where $M_{a \vert x} = M^{(1)}_{a_1 \vert x_1} \otimes \ldots \otimes M^{(n)}_{a_n \vert x_n}$.
    In particular, for all $x \in \Z_2^n, \bop_x(M)\rho = \beta_{\tilde a(x), x} \rho$.
\end{lemma}
\begin{proof}
    First suppose that $\rho = \ketbra{\psi}{\psi}$ is a pure quantum state.
    First, by definition of $p_{\tilde a}$, we know that for all $x \in \Z_2^n$: 
    \begin{align}
        \pmap(\Sigma)(\tilde a_1(x_1) \ldots \tilde a_n(x_n) \vert x) &= 1\\
        \vert \bra{\psi} M^{(1)}_{\tilde a_1(x_1) \vert x_1} \otimes \ldots \otimes M^{(n)}_{\tilde a_n(x_n) \vert x_n} \ket \psi \vert^2 &= 1\\
        \vert \vert M^{(1)}_{\tilde a_1(x_1) \vert x_1} \otimes \ldots \otimes M^{(n)}_{\tilde a_n(x_n) \vert x_n} \ket \psi \vert \vert^2 &\ge 1 \qquad \text{by Cauchy-Schwarz} 
    \end{align}

Now, note that since $M_{\tilde{a}(x) \vert x} \equiv M^{(1)}_{\tilde a_1(x_1) \vert x_1} \otimes \ldots \otimes M^{(n)}_{\tilde a_n(x_n) \vert x_n}$ is a POVM element, its maximal eigenvalue $\lambda_{\max}$ must be lower than $1$. 
Hence, the last inequality forces $\lambda_{\max} = 1$ and $M_{\tilde{a}(x) \vert x} \ket{\psi}= \ket{\psi}$.
For a general density operator $\rho = \sum_{j = 1}^k p_j \ketbra{\psi_j}{\psi_j}$, with $\sum_{j = 1}^k p_j = 1$,we have:
\begin{equation}
     \pmap(\Sigma)(\tilde a_1(x_1) \ldots \tilde a_n(x_n) \vert x) = \sum_{j = 1}^k p_j \vert \bra{\psi_j} M_{\tilde a (x) \vert x} \ket{\psi_j} \vert^2 = 1,
\end{equation} 
implying $\vert \bra{\psi_j} M_{\tilde a (x) \vert x} \ket{\psi_j} \vert^2 = 1$ for all $1 \le j \le k$, which by the previous discussion concludes.
\end{proof}
The following energy-weighted sum rule will be crucial in the study of the response of $\vec \beta$ under perturbations of $\Sigma$. It 
forms a connection between the spectral structure of $A$ and the evolution under $X$ of $A$ up to second order when the 
starting point is an eigenvector of $A$. 

\begin{lemma}[Second order energy weighted sum rule for observables]
    Let $A$ be any observable, and $\lambda_1, \ldots, \lambda_N$, $A_1, \ldots, A_N$ be such that $A = \sum_{m = 1}^N \lambda_m A_m$ and
    $\sum_{m = 1}^N A_m = I$. Let $\rho$ be a density operator such that $A \rho = \rho A = \lambda \rho$ for some eigenvalue $\lambda$ of $A$.
    Let $X$ be skew-Hermitian.
    \begin{align}
        \label{eq:criticalpoint}
        \Tr([X, A]\rho) &= 0\\
        \label{eq:fermi}
        \Tr([X, A]_2\rho) &= 2\Tr((A - \lambda)X\rho X^\dagger) = 2\sum_{m = 1}^N (\lambda_m - \lambda) \Tr(A_m X\rho X^\dagger).
    \end{align}
    This is in particular the case if $A_1, \ldots, A_N$ are the orthogonal projectors on the eigenspaces of $A$.
\end{lemma}

\begin{proof}
    We simply develop the commutators:
    \begin{align}
        \Tr([X, A]\rho) &= \Tr(XA\rho) - \Tr(AX \rho) \\
        &= \Tr(XA \rho) - \Tr(\rho A X) \\
        &= \lambda(\Tr(X\rho) - \Tr(\rho X)) = 0\\
        \Tr([X, A]_2 \rho) &= \Tr((X^2A - 2XAX - AX^2)\rho)\\
        &= \Tr(\lambda X^2\rho) - 2\Tr(XAX\rho) + \Tr(\lambda X^2)\\
        &= 2(\Tr(\lambda X\rho X) - \Tr(AX \rho X))\\
        &=2\Tr((A - \lambda I)X\rho X^\dagger) && X^\dagger = -X\\
        &=2\Tr(\left(\sum_{m = 1}^N \lambda_m A_m  - \lambda \sum_{m = 1}^N A_m \right) X\rho X^\dagger)\\
        &=2\sum_{m = 1}^N (\lambda_m - \lambda) \Tr(A_m X\rho X^\dagger).
    \end{align}

\end{proof}
Translating to our setting, this yields the following formula.

\begin{lemma}
    \label{lem:secondorder_det}
    Let $\Sigma$ be a quantum strategy realizing a local deterministic correlation $p_{\tilde a}$. Then $\Sigma$ is a critical point of 
    its unitary orbit for any Bell functional $\vec \beta$. More precisely, if we denote $\vec U = (e^{K_1}, \ldots, e^{K_n}, e^{-K_\rho})$, and 
    $\Pi_{a, x}$ the projector on the eigenspace of $\bop_x(M)$ with eigenvalue $b_{a, x}$, then: 
    \begin{equation}
        \vec \beta \cdot \pmap(\vec U \cdot \Sigma) = \vec \beta \cdot \pmap(\Sigma) + \sum_{x, a} (b_{a, x} - \beta_{\tilde a(x), x})\Tr(\Pi_{a, x} (x \cdot \vec K + K_{\rho}) \rho (x \cdot \vec K + K_{\rho})^\dagger) + o(\vert \vert \vec K \vert \vert^2)
    \end{equation}
\end{lemma}
\begin{proof}
    For all $x \in \Z_2^n$, since $\Sigma$ is local deterministic, we have $\bop_x(M)\rho = \rho \bop_x(M) = \beta_{\tilde a (x), x}\rho$.
    Thus, by (\ref{eq:criticalpoint}), the first order term $\Tr([x\cdot \vec K + K_\rho, \bop_x(M)]\rho)$ of (\ref{eq:second_order_generic}) vanishes, as well as the 
    second order term $[[x\cdot \vec K + K_\rho], \bop_x(M)]$. We obtain the desired result by applying (\ref{eq:fermi}) to $[x\cdot \vec K + K_\rho, \bop_x(M)]_2$.
\end{proof}

We now turn to the specific case of quantum projective strategies.

In this case, the eigenvalues of $\bop_x(M)$ can be 
directly related to $\vec \beta$, as they are exactly the numbers $\beta_{a, x}, a \in \Z_d^n$. The unitary response of $\vec \beta$ can be split into a main term 
depending on the evolved state $K_\rho \rho K_\rho^\dagger$, and a remaining term involving 
measurement basis variations, which has support only on outcomes differing from $\tilde a(x)$ on only one site.
\begin{definition}(One-flip maximum)
    A local deterministic correlation $p_{\tilde a}$ is a one-flip maximum (resp. strict one-flip maximum) for $\vec \beta$ if and only if for all $1 \le j \le n$, and for all $\tilde a' = (\tilde a'_k)_{1 \le k \le n}$
    such that $\tilde a'_k = \tilde a_k$ for all $k \neq j$, $\vec \beta \cdot p_{\tilde a} \ge \vec \beta \cdot p_{\tilde a'}$ (resp. $\vec \beta \cdot p_{\tilde a} > \vec \beta \cdot p_{\tilde a'}$).
\end{definition}
Note that being a one-flip maximum is a necessary condition for classical optimality.

\begin{lemma}(Limited effect of measurement evolution for QP strategies)
    \label{lem:secondorder_qp}
    Let $\Sigma = (\rho, \Pi)$ be a QP strategy of local dimension $D \ge d$. For $a\in \Z_d^n$, $x \in \Z_2^n$, denote $w(a) \coloneq |\{1 \le i \le n \mid a_i \neq 0\}|$, and $\Pi_{a \vert x} \equiv \Pi^{(1)}_{a_1 \vert x_1} \otimes \ldots \otimes \Pi^{(n)}_{a_n \vert x_n}$. Then, using the same notation as in Proposition \ref{prop:seriesexpgen}:
    \begin{equation}
        \vec \beta \cdot \pmap(\vec U \cdot \Sigma) = \vec \beta \cdot \pmap(\Sigma) + \Tr(\bop_2(\Pi) K_\rho \rho K_\rho^\dagger) + R(\Pi, \vec K, \rho)+ o(\vert \vert \vec K \vert \vert^2)
    \end{equation}
    where $R(\Pi, \vec K, \rho) \equiv \sum_{x, a:w(a-\tilde{a}(x)) = 1} (\beta_{a, x} - \beta_{\tilde{a}(x), x})\Tr(\Pi_{a \vert x}(x_jK_j + K_\rho)\rho(x_jK_j + K_\rho)^\dagger)$ only involves ouptut strings with $w(a-\tilde{a}(x)) = 1$, and $\bop_2(\Pi) = \sum_{x, a:w(a-\tilde{a} (x)) \ge 2}(\beta_{a, x} - \beta_{\tilde{a}(x), x})\Pi_{a \vert x}$
    
    Moreover, if $D = d$ and $p_{\tilde a}$ is a one-flip maximum, then $R(\Pi, \vec K, \rho) \le 0$.
\end{lemma}
\begin{proof}
    Recall that for a QP strategy, $\bop_x(M) = \sum_{a \in \Z_d^n} \beta_{a, x} \Pi_{a\vert x}$, which is a spectral decomposition of $\bop_x(M)$, since 
    the projectors $\Pi_{a \vert x}$ partition the identity for a fixed $x$. Thus, we can use Lemma \ref{lem:secondorder_det} with $\Pi_{a, x} = \Pi_{a \vert x}$, which yields, up to second order:
    \begin{align}
        &\vec \beta \cdot \pmap(\vec U \cdot \Sigma) - \vec \beta \cdot \pmap(\Sigma) = \sum_{x, a} (\beta_{a, x} - \beta_{\tilde a (x), x})\Tr(\Pi_{a \vert x}(x \cdot \vec K + K_{\rho}) \rho (x \cdot \vec K + K_{\rho})^\dagger)\\
        &= \sum_{x, a} (\beta_{a, x} - \beta_{\tilde a (x), x})\Tr(\Pi_{a \vert x}(x \cdot \vec K + K_{\rho}) \rho (x \cdot \vec K + K_{\rho})^\dagger\Pi_{a\vert x}) && \Pi_{a\vert x} = \Pi_{a \vert x}^2\\
        &= \sum_{x, a} (\beta_{a, x} - \beta_{\tilde a (x), x})\times \nonumber \\
      & \quad \Tr((\Pi_{a \vert x}(x\cdot \vec K)\Pi_{\tilde a(x)\vert x} + \Pi_{a \vert x} K_{\rho}\Pi_{\tilde a(x) \vert x}) \rho (\Pi_{a \vert x}(x\cdot \vec K )\Pi_{\tilde a(x) \vert x} + \Pi_{a \vert x} K_{\rho}\Pi_{\tilde a(x) \vert x})^\dagger)
    \end{align}
    Where in the last line, we used $\Pi_{\tilde a(x) \vert x}\rho = \rho \Pi_{\tilde a(x) \vert x} = \rho$. 
    Note that most of the contributions of $x \cdot \vec K$ are null. Indeed, for $a \in \Z_d^n, x \in \Z_2^n$: 
    \begin{align}
        \label{eq:factorized_measurement_variation}
        \Pi_{a \vert x} (x \cdot \vec K) \Pi_{\tilde a(x) \vert x} &= \sum_{i = 1}^n x_i \bigotimes_{j = 1}^n \Pi^{(j)}_{a_j \vert x_j}K_i^{\delta_{ij}}\Pi^{(j)}_{\tilde a_j(x_j) \vert x_j}
    \end{align}
    Note that $\Pi^{(j)}_{b \vert x_j}\Pi^{(j)}_{b' \vert x_j} = 0$ for $b \neq b'$. If $w(a - \tilde a(x)) \ge 2$, at least one of the overlaps $\Pi^{(j)}_{a_j \vert x_j}\Pi^{(j)}_{\tilde a_j(x_j) \vert x_j}$ where $a$ and $\tilde a(x)$ differ at $j$ appear, making the overlap equal to $0$. Hence, if $w(a - \tilde a(x)) \ge 2$, $\Pi_{a \vert x} (x \cdot \vec K) \Pi_{\tilde a(x) \vert x} = 0$.
    The strings $a$ such that $w(a) = 1$ are denoted by $k^j$, for $1 \le k \le d$, $1 \le j \le n$, where $k^j$ is the string equal to $0$ on all sites but on $j$ where $k^j_j = k$.
    Consequently, the second order term splits as: 
    \begin{align}
        \sum_{x, a:w(a - \tilde a) \ge 2} (\beta_{a, x} - \beta_{\tilde a (x), x})\Tr(\Pi_{a \vert x}K_\rho \rho K_{\rho}^\dagger) + R(\Pi, \vec K, \rho)) = \Tr(\bop_2(\Pi) K_\rho \rho K_\rho^\dagger) + R(\Pi, \vec K, \rho))
    \end{align}
    where $R(\Pi, \vec K, \rho))\equiv  \sum_{k, j} \sum_x (\beta_{\tilde a(x) + k^j, x} - \beta_{\tilde a(x), x})\Tr(\Pi_{\tilde a(x) + k^j \vert x} (x_jK_j + K_{\rho})\rho(x_jK_j + K_{\rho})^\dagger)$, as announced. 

    Suppose now that $D = d$. Then, since $\Pi_{\tilde a (x) \vert x}$ is of rank $1$, 
    $\Pi_{\tilde a (x) \vert x} \rho = \rho \Pi_{\tilde a(x)\vert x} = \rho$ implies:
    
    \begin{equation}
    \label{eq:rhoispi}
    \rho = \Pi_{\tilde a(x) \vert x} = \bigotimes_{j = 1}^n \Pi^{(j)}_{\tilde a_j(x_j)\vert x_j}, 
    \end{equation}
    independently of $x \in \Z_d^n$.
    In particular, for all $1 \le j \le n$, $\Pi^{(j)}_{\tilde a_j(1)\vert 1} = \Pi^{(j)}_{\tilde a_j(0)\vert 0}$. This allows us to write :

    \begin{align}
        \label{eq:R}
        R(\Pi, \vec K, \rho) &= \Tr(\sum_{k, j} \bigotimes_{i = 1}^{j-1} \sum_{x} (\beta_{\tilde a(x) + k^j, x} - \beta_{\tilde a(x), x}) \Pi^{(i)}_{\tilde a_i(0) \vert 0} \otimes \Pi^{(j)}_{\tilde a_j(x_j) + k \vert x_j} \otimes \bigotimes_{i = j+1}^n \Pi^{(i)}_{\tilde a_i(0) \vert 0} (x_jK_j + K_{\rho})\rho(x_jK_j + K_{\rho})^\dagger)\\
        &= \sum_{k, j} \Tr(R_{k, j, 0} K_{\rho} \rho K_{\rho}^\dagger) + \sum_{k, j} \Tr(R_{k, j, 1} (K_j + K_{\rho}) \rho (K_j + K_{\rho})^\dagger)
    \end{align}
    
    Where we set $R_{k, j, b} = \bigotimes_{i = 1}^{k-1} \Pi^{(i)}_{\tilde a_i(0) \vert 0} \otimes \left(\sum_{x : x_j = b} \beta_{\tilde a(x) + k^j, x} - \beta_{\tilde a(x), x}\right) \Pi^{(j)}_{\tilde a_j(x_j) + k \vert x_j} \bigotimes_{i = k+1}^{n} \Pi^{(i)}_{\tilde a_i(0) \vert 0}$.
    The hypothesis that $p_{\tilde a}$ is a one flip maximum (resp. strict one flip maximum) is equivalent to $R_{k, j, b} \le 0$ (resp. $R_{k, j, b} < 0$) for all $k \in \Z_d, j \in \{1, \ldots, n\}, b \in \{0, 1\}$, which implies $R \le 0$.

    The reciprocal does not necessarily hold, but note that for $(k, j) \neq (k', j')$, $R_{k, j, b}R_{k', j', b'} = 0$, so $R < 0$ implies $R_{k, j, 0} + R_{k, j, 1} < 0$ for all $k, j$ (take $\vec K, \rho)= 0$ and $K_\rho$ arbitrary, which allows $K_\rho \rho K_\rho^\dagger$ to describe any pure state).

\end{proof}

This already highlights a difference between PVMs and POVMs: for PVMs, the effect of conditional measurement variations is very limited, as they only act at second order through the projectors $\Pi_{\tilde a(x) \vert x + k^j \vert x}$, and not at all if $D = d$ and $p_{\tilde a}$ is a one-flip maximum.
For POVMs there is no such restriction a priori. This is due to the overlaps between POVM elements, which forbid to conclude that most terms in (\ref{eq:factorized_measurement_variation}) vanish.

We now prove a dimension reduction result for deciding if a given local deterministic strategy is a local optimum in $\qp(d)$, showing how structured the optimization problem in $\qp(d)$ is 
compared to $\qset(d)$. In short, we 
show that this problem reduces to finding a $\qp(d-1)$ strategy with positive gain for a set of non-local games with $d-1$ outcomes constructed from $\vec \beta$.

\begin{lemma}[Dimension reduction via subset games]
    \label{lem:dimred}
    For all $k \in \mathbb{N}$, $c \in \Z_k^n$, define the support of $c$ as $S(c) = \{1 \le i \le n \mid c_i \neq 0\}$, and 
    for $S = \{s_1, \ldots, s_r\} \subseteq \{1, \ldots, n\}$, denote the restriction of $c$ to $S$ as $c_S = c_{s_1}\ldots c_{s_n}$.
    
    For $S = \{s_1, \ldots, s_r \} \subseteq \{1, \ldots, n\}$, $x \in \Z_2^n$, $\alpha \in \Z^{S, x}_d \equiv \{\alpha \in \Z_d^{|S|} \mid \forall 1 \le i \le |S|, \alpha_i \neq \tilde a_{s_i}(x_{s_i})\}$, denote $\alpha^{x, S}$ the unique string in $\Z_d^n$ such that 
    $\alpha^{x, S}_S = \alpha$, and $\alpha^{x, S}_j = \tilde a_j(x_j)$ for all $j \notin S$. Then, define, for all $\xi \in \Z_2^{|S|}$: 
    
    \begin{equation}
        \beta^S_{\alpha, \xi} \coloneq  \sum_{x : x_S = \xi}\beta_{\alpha^{x, S}, x} - \beta_{\tilde a(x), x}.
    \end{equation}
    Up to relabeling of the output strings, $\vec \beta^S$ defines a $|S|2(d-1)$ functional, to which one can associate 
    a Bell operator $\bop^S(M) \coloneq \sum_{\alpha, \xi}\beta^S_{\alpha, \xi} M_{\alpha \vert \xi}$, which we call a subset game.
    Let $\Sigma = (\rho, \Pi) \in \qp(d)$ such that $\pmap(\Sigma) = p_{\tilde a}$. Then, $\bop_2(M)$ admits a direct sum decomposition into subset games.
    
    \begin{equation} 
        \bop_2(\Pi) = \sum_{\substack{S \subseteq \{1, \ldots, n\} \\ |S| \ge 2}} I_S \otimes \bop^S(\Pi^S),
    \end{equation}
    where $I_S = \bigotimes_{j \notin S} \Pi^{(j)}_{\tilde a_j(0) \vert 0}$ is the projector used by the parties that are inactive for the subset $S$, and $\Pi^S = (\Pi^{(s_1)}, \ldots, \Pi^{(s_{|S|})})$ is the measurement strategy of the players in $S$.
\end{lemma}

In other words, to investigate second order variations of $\vec \beta \cdot p$ around a local deterministic correlation $p_{\tilde a}$,
one can equivalently study each of the $(k, 2, d-1)$ functional defined by the average gain that the subset \( S \) of $1 \le k \le n$ players obtains by changing their projective strategies —forbidding projectors that can yield the outputs
prescribed by \( \tilde a(x) \).

\begin{proof}
    Let $x, x' \in \Z_2^n$, and $a, a' \in \Z_d^n$ such that $S(a - \tilde a(x)) \neq S(a' - \tilde a(x'))$. Then, there exists
    $1 \le j \le n$ such that $a_j = \tilde a_j(x_j)$ and $a'_j \neq \tilde a_j(x'_j)$. Recall that $\pmap(\Sigma) = p_{\tilde a}$ forces
    $\Pi_{\tilde a_j(0) \vert 0} = \Pi_{\tilde a_j(1) \vert 1}$.
    Thus, $\Pi_{a \vert x} \Pi_{a' \vert x} = \bigotimes_{i = 1}^n \Pi^{(i)}_{a_i \vert x_i} \Pi^{(i)}_{a'_i \vert x'_i} = 0$ because $\Pi_{\tilde a_j \vert x_j} \Pi_{\tilde a'_j\vert x'_j} = \Pi_{\tilde a_j(x_j)\vert x_j} \Pi_{\tilde a'_j \vert x'_j} = \Pi_{\tilde a_j(x'_j) \vert x'_j}\Pi_{a'_j \vert x'_j} = 0$.
    
    We deduce the following direct sum decomposition for $\bop_2(\Pi)$, sorting the output strings $a$ 
    by their support $S$ (where they do not match with $\tilde a(x))$:
    \begin{equation}
        \label{eq:supportdec}
        \bop_2(\Pi) \coloneq\sum_{x, a:w(a-\tilde{a} (x)) \ge 2}(\beta_{a, x} - \beta_{\tilde{a}(x), x})\Pi_{a \vert x} = \sum_{\substack{S \subseteq \{1, \ldots, n\} \\ |S| \ge 2}} \sum_{\xi \in \Z_2^{|S|}}\sum_{\alpha \in \Z^{S, x}_d} \sum_{x : x_S = \xi} (\beta_{\alpha^{x, S}, x} - \beta_{\tilde a(x), x})\Pi_{\alpha^{x, S} \vert x},
    \end{equation}
    where the terms in the first sum are in direct sum by the previous remark. Observe that up to to permutation of the tensor product, we have:
    \begin{align}
    \Pi_{\alpha^{x, S} \vert x} &= \bigotimes_{i \in S} \Pi^{(i)}_{\alpha^{x, S}_i \vert x_i}\bigotimes_{i \notin S} \Pi^{(i)}_{\alpha^{x, S}_i \vert x_i}\\
    &= \bigotimes_{i = 1}^{|S|}\Pi^{(s_i)}_{\alpha_i \vert \xi_i} \bigotimes_{i \notin S} \Pi^{(i)}_{\tilde a_i(x_i) \vert x_i} \\
    &= \bigotimes_{i = 1}^{|S|} \Pi^{(s_i)}_{\alpha_i \vert \xi_i} I_S,
    \end{align}
    where $I_S \equiv \bigotimes_{i \notin S} \Pi^{(i)}_{\tilde a_i(x_i) \vert x_i} =  \bigotimes_{i \notin S} \Pi^{(i)}_{\tilde a(0) \vert 0}$ because $\pmap(\Sigma) = p_{\tilde a}$.

    Reinjecting into (\ref{eq:supportdec}), we obtain, again up to a permutation of the factors depending only on $S$:
    \begin{align}
        \bop_2(\Pi) &= \sum_{\substack{S \subseteq \{1, \ldots, n\} \\ |S| \ge 2}} I_S \otimes \sum_{\xi \in \Z_2^{|S|}}\sum_{\alpha \in \Z^{S, x}_d} \sum_{x : x_S = \xi} (\beta_{\alpha^{x, S}, x} - \beta_{\tilde a(x), x})\bigotimes_{i = 1}^{|S|} \Pi^{(i)}_{\alpha_i \vert \xi_i} \\
        &= \sum_{\substack{S \subseteq \{1, \ldots, n\}\\ |S| \ge 2}} I_S \otimes \sum_{\xi, \alpha} \beta^{S}_{\alpha, \xi} \bigotimes_{i = 1}^{|S|} \Pi^{(s_i)}_{\alpha_i \vert \xi_i}\\
        &= \sum_{\substack{S \subseteq \{1, \ldots, n\}\\ |S| \ge 2}} I_S \otimes \bop^S(\Pi^S),
    \end{align}
    where we denoted $\Pi^S = (\Pi^{(s_1)}, \ldots, \Pi^{(s_{|S|})})$ the measurement strategy of the players in $S$, and $\bop^S$ the Bell operator associated to 
    $\vec \beta^S$.
\end{proof}

If none of these games admit a $\qp(d-1)$ strategy with positive gain, then $p_{\tilde a}$ is 
a local optimum in $\qp(d)$ for $\vec \beta$. 

\begin{corollary}[Local optimality criterion via subset games]
    \label{cor:subsetneg}
    For $S \subseteq \{1, \ldots, n\}$, defsine the active subspace of the subset game $\bop_S$ as $V_S \coloneq \bigoplus_{\substack{(\alpha, \xi) \in A_S}} \Im(\Pi_{\alpha \vert \xi}) $,
    with $A_S \coloneq \{(\alpha, \xi), \alpha \in \Z_d^{|S|}, \xi \in \Z_2^{|S|} \mid \forall 1 \le i \le n, \alpha_i \neq \tilde a_i(\xi_i)\}$.
    Suppose that $p_{\tilde a}$ is a strict one-flip maximum, and that for all $S \subseteq \{1, \ldots, n\}$ with $|S| \ge 2$, and all quantum projective strategy $\Sigma = (\Pi, \rho)$ of local dimension $d$ such that $\pmap(\Sigma) = p_{\tilde a}$, $\bop_S(\Pi^S)\vert_{V_S} \prec 0$. Then, $p_{\tilde a}$ is a local optimum in $\qp(d)$.
    Moreover, if there exists such a $\Sigma$ and $S$, $\ket{\psi_S} \in V_S$ such that $\bra{\psi_S}\bop_S(\Pi^S)\ket{\psi_S} > 0$, then $p_{\tilde a}$ is not a local optimum in $\qp(d)$.
\end{corollary}

This property is not obvious for two reasons. First, it requires to write the neighborhood of $p_{\tilde a}$ from the union of the neighborhood
of its preimages, which requires a crucial compacity hypothesis. Second, we need to deal with degenerate directions where the second order term is zero, and guarantee that 
higher order terms are not positive.

\begin{lemma}
      \label{lem:neighborhoodimages}
    Consider two separated topological spaces $X$ and $Y$, with $X$ compact. 
    Let $f : X \rightarrow Y$ be a continuous map.
    Then, for all $y \in Y$, $\bigcup_{x \in f^{-1}(y)} f(V_x)$ (where $V_x$ denotes any neighborhood of $x$), is a neighborhood of $y$ in $f(X)$. 
\end{lemma}
\begin{proof}
    Note that $Z = X - \bigcup_{x \in f^{-1}(y)} V_x$ is compact, since it is closed in the compact $X$. Hence, $f(Z)$ is compact, and thus closed since $Y$ is separated, and does not contain $y$, implying that $f(X) - f(Z)$ is an open set, which contains $y$ since $f(X)$ does.
    To conclude, notice that $f(X) - f(Z) \subseteq f(X - Z) = f(\bigcup_{x \in f^{-1}(y)} V_x) = \bigcup_{x \in f^{-1}(y)} f(V_x)$.
\end{proof}

\begin{proof}[Proof of Corollary \ref{cor:subsetneg}]
    Let $\Sigma = (\rho ,\Pi)$ be a $\qp(d)$ strategy such that $\pmap(\Sigma) = p_{\tilde a}$. For all $\varepsilon > 0$, the set 
    \begin{equation}
        B_{\varepsilon} \equiv \{\vec U(\vec K) \cdot \Sigma, \vert \vert \vec K \vert \vert < \varepsilon \},
    \end{equation} 
    where $\vec U (\vec K) = (e^{K_1}, \ldots, e^{K_n}, e^{-K_\rho})$, is a neighborhood of $\Sigma$ in $\qp(d)$, which is compact. 
    Hence, by Lemma $\ref{lem:neighborhoodimages}$, it suffices to show that for $\vert \vert \vec K \vert \vert$ small enough, 
    $\vec \beta \cdot \pmap(\vec U(\vec K) \cdot \Sigma) < \vec \beta \cdot \pmap(\Sigma)$. Recall Lemma \ref{lem:secondorder_qp}: 
    
    \begin{equation}
        \vec \beta \cdot \pmap(\vec U \cdot \Sigma) = \vec \beta \cdot \pmap(\Sigma) + \Tr(\bop_2(\Pi) K_\rho \rho K_\rho^\dagger) + R(\Pi, \vec K, \rho))+ o(\vert \vert \vec K \vert \vert^2).
    \end{equation}
    By Lemma \ref{lem:secondorder_qp}, since $p_0$ is a one-flip maximum, $R(\Pi, \vec K, \rho)) \le 0$. By Lemma \ref{lem:dimred}, since $B^S(\Pi^S) \prec 0$ for all subset of parties $S$, we deduce that $\Tr(\bop_2(\Pi)K_\rho \rho K_\rho^\dagger) \le 0$.
    Now suppose that the second order term is zero. Then, we must have $R(\Pi, \vec K, \rho)) = 0$ and $\Tr(\bop_2(\Pi)K_\rho \rho K_\rho^\dagger) = 0$. 
    Since $\bop_2(\Pi) \le 0$, the latter implies that $\bop_2(\Pi)K_\rho \rho K_\rho^\dagger = 0$, so the support of $K_\rho \rho K_\rho^\dagger$ is in $\ker \bop_2(\Pi)$. 
        
    Since the PVMs in $\Pi$ are not degenerate, let us denote $\Pi_{a \vert x} = \ketbra{a \vert x}{a \vert x}$ and show that
    
    \begin{equation}
        \ker \bop_2(\Pi) = \vect(\ket{a \vert 0^n}, w(a - \tilde a(0^n)) \le 1).
    \end{equation}
    Let $\ket{\psi} \in \ker \bop_2(\Pi)$, and write $\ket{\psi} = \sum_{a \in \Z_d^n} c_a \ket{a \vert 0^n}$. 
    By definition $\bop_2(\Pi)\ket{\psi} = 0$, but since  by hypothesis $I_S \otimes \bop_S(\Pi^S) \le 0$ for all $S \subseteq \{1, \ldots, n\}$, we deduce that $I_S \otimes \bop_S(\Pi^S)\ket{\psi} = 0$. Let us compute the latter: 
    \begin{equation}
        I_S \otimes \bop_S(\Pi^S)\ket{\psi} = \sum_{a \in \Z_d^n} c_a I_S \otimes \bop_S(\Pi^S)\ket{a \vert 0^n} = \sum_{a \in \Z_d^n} c_a \left(I_S \bigotimes_{i \notin S} \ket{a_i \vert 0} \right) \otimes \left(\bop_S(\Pi^S) \bigotimes_{i \in S} \ket{a_i \vert 0}\right). 
    \end{equation}
    Now, by definition of $I_S$ and $B_S(\Pi^S)$, $I_S \bigotimes_{i \notin S} \ket{a_i \vert 0}$ is non zero if and only if $a_i = \tilde a_i(0)$ for all $i \notin S$, and $\bop_S(\Pi^S) \bigotimes_{i \in S} \ket{a_i \vert 0}$ is non zero if and only if $a_i \neq \tilde a_i(0)$ for $i \in S$.
    Hence, $I_S \otimes \bop_S(\Pi^S)\ket{a \vert 0^n} \neq 0$ if and only if $S(a - \tilde a(0^n)) = S$. Hence, the last equation simplifies to: 
    \begin{equation}
        I_S \otimes \bop_S(\Pi^S) \ket{\psi} = \bigotimes_{i \notin S} \ket{\tilde a_i(0)\vert 0} \otimes \bop_S(\Pi^S)\ket{\psi_S},
    \end{equation}
    where $\ket{\psi_S} = \sum_{a : S(a - \tilde a(0^n)) = S} c_a \ket{a_S}$. Since $\bop_S(\Pi^S)$ is strictly negative on $V_S$ and $\ket{\psi_S} \in V_S$ by 
    definition, we must have $\ket{\psi_S} = 0$, and thus $c_a = 0$ for all $a$ such that $S(a - \tilde a(0^n)) = S$. Since this is true for all $S$ with $|S| \ge 2$,
    only the coefficients $c_a$ with $w(a - \tilde a(0^n)) \le 1$ remain, which proves that $\ker \bop_2(\Pi) = \vect(\ket{a \vert 0^n}, w(a - \tilde a(0^n)) \le 1)$.
    
    We now use that $R(\Pi, \vec K, \rho) = 0$, which rewrites, using the same notation as in (\ref{eq:R}): 
    \begin{equation}
        \label{eq:Rdec}
        \sum_{k, j} \Tr(R_{k, j, 0} K_{\rho} \rho K_{\rho}^\dagger) + \sum_{k, j} \Tr(R_{k, j, 1} (K_j + K_{\rho}) \rho (K_j + K_{\rho})^\dagger) = 0.
    \end{equation}
    Since $p_{\tilde a}$ is a one-flip maximum, $R_{k, j, b}$ is negative, meaning the support of $K_\rho \rho K_\rho^\dagger$ must also lie in the kernel of $R_{k, j, b}$. 
    Moreover, since the optimum is strict, the factor of $R_{k, j, b}$ on the $k$-th site, $\sum_{x : x_j = b} (\beta_{\tilde a(x) + k^j, x} - \beta_{\tilde a(x), x}) \Pi_{\tilde a_k(x) \vert b}$, is 
    strictly negative on the subspace $V_{\{k\}}$. By an analogous reasoning, taking $b = 0$, we deduce that the coefficients $c_a$ of a state $\ket{\psi}$ in the support of $K_\rho \rho K_\rho^\dagger$ for $w(a - \tilde a(0)) = 1$ 
    must be zero, hence $K_\rho \rho K_\rho^\dagger \propto \ketbra{\tilde a(0^n) \vert 0^n}{\tilde a(0^n) \vert 0^n}$, which means $K_\rho \rho \propto \ketbra{\tilde a(0^n) \vert 0^n}{\phi}$ for some state $\ket{\phi}$. Now note that since $\pmap(\Sigma) = p_{\tilde a}$ and 
    the PVMs are not degenerate, we must have $\rho = \Pi_{\tilde a(0^n) \vert 0^n}$ by Lemma $\ref{lem:localdeterministic}$. Combined with the previous, we conclude that $K_\rho \rho \propto \rho$.
    We now turn to the second term of $(\ref{eq:Rdec})$. Again, since $R_{k, j, 1} \preceq 0$ for all $k, j$, the summand is null for all $k, j$.
    Note that $R_{k, j, 1}\rho = \rho R_{k, j, 1} = 0$, and hence, denoting $K_\rho \rho = i \lambda \rho$ for some real $\lambda$: 
    \begin{align}
        R_{k, j, 1}&(K_j + K_\rho) \rho (K_j + K_\rho)^\dagger = R_{k, j, 1}(K_j \rho K_j^\dagger  + \lambda (\rho K_j^\dagger + K_j \rho) + K_\rho \rho K_\rho^\dagger) = R_{k, j, 1} K_j \rho K_j^\dagger + \lambda R_{k, j, 1} K_j \rho\\
        0 &= \Tr(R_{k, j, 1}(K_j + K_\rho) \rho (K_j + K_\rho)^\dagger) = \Tr(R_{k, j, 1} K_j \rho K_j^\dagger) + \lambda \Tr(\rho R_{k, j, 1} K_j) \\
        &=\Tr(R_{k, j, 1} K_j \rho K_j^\dagger),
    \end{align}
    where we used trace cyclicity in the second to last line. Again, this induces $R_{k, j, 1} K_j \rho K^\dagger_j = 0$, and thus: 
    \begin{equation}
        0 = R_{k, j, 1} K_j \rho K^\dagger_j = \bigotimes_{i = 1}^{j-1} \Pi^{(i)}_{\tilde a_j(0) \vert 0} \otimes \left(\sum_{x: x_j = 1} (\beta_{\tilde a(x) + k^j, x} - \beta_{\tilde a(x), x}) \Pi^{(j)}_{\tilde a_j(1) + k \vert 1}K_j \Pi_{\tilde a_j(0) \vert 0}K_j^\dagger\right) \otimes \bigotimes_{i = j+1}^{n} \Pi^{(i)}_{\tilde a_i(0) \vert 0}.
    \end{equation}
    Since $p_{\tilde a}$ is a strict one-flip maximum, $\sum_{x: x_j = 1} (\beta_{\tilde a(x) + k^j, x} - \beta_{\tilde a(x), x}) < 0$, so $\Pi_{\tilde a_j(1) + k \vert 1}K_j \Pi_{\tilde a_j(0) \vert 0}K_j^\dagger = 0$. This is true for all $1 \le k \le d-1$, but 
    the projectors $\Pi_{\tilde a_j(1) + k \vert 1}K_j \Pi_{\tilde a_j(0)}, k \in \Z_d$ form a partition of the identity, so $K_j \Pi_{\tilde a_j(0) \vert 0}K_j^\dagger \propto \Pi_{\tilde a_j(0) \vert 0}$, which is equivalent to $K_j \rho K_j^\dagger \propto \rho$,
    
    We deduce that for all $x \in \Z_2^n$, $(x \cdot \vec K + K_\rho) \rho (x \cdot \vec K + K_\rho)^\dagger \propto \rho$. By the same reasoning as earlier, we deduce that $(x \cdot \vec K + K_\rho) \rho = i\mu_x \rho = \rho (x \cdot \vec K + K_\rho)$ for some real $\mu_x$.
    Finally, we deduce that if the second order term is not strictly negative, then: 
    \begin{equation}
        \vec \beta \cdot \pmap(\vec U \cdot \Sigma) = \Tr(\sum_{x} \bop_x(\Pi) e^{x \cdot \vec K + K_\rho} \rho e^{-x \cdot \vec K - K_\rho}) = \Tr(\sum_x \bop_x (\Pi) e^{i\mu_x}\rho e^{-i\mu_x}) = \vec \beta \cdot \pmap(\Sigma),
    \end{equation}
    which concludes.
\end{proof}

%% file: sections_journal_v3/section5.tex
\section{The (n, 2, 2) quantum set is flat around every local deterministic correlation}
We conclude with an unexpected consequence for the geometry of $\qset$ in the $(n, 2, 2)$ scenario, namely that the faces 
around any local deterministic point are flat. We make this statement formal in the following theorem.

\begin{theorem}
    \label{th:flat}
    Let $p_{\tilde a}$ be a local deterministic correlation, and $P$ be a $2$-dimensional plane such that $p_{\tilde a} \in \qset \cap P$.
    Then, there exists $r_P > 0$ such that $\ext(Q) \cap P \cap B(p_{\tilde a}, r_P) = \{p_{\tilde a}\}$.
\end{theorem}

We first show the following proposition based on the dimension reduction lemma from the previous section.
\begin{lemma}
    \label{lem:locopt}
    Let $p_{\tilde a}$ be a local deterministic correlation which is a unique classical maximum for a $(n, 2, 2)$ Bell functional $\vec \beta$.
    Then, $p_{\tilde a}$ is a local optimum for $\vec \beta$ in $\qp(2)$.
\end{lemma}

\begin{figure}[h]
    \centering
    \includegraphics[width=\textwidth]{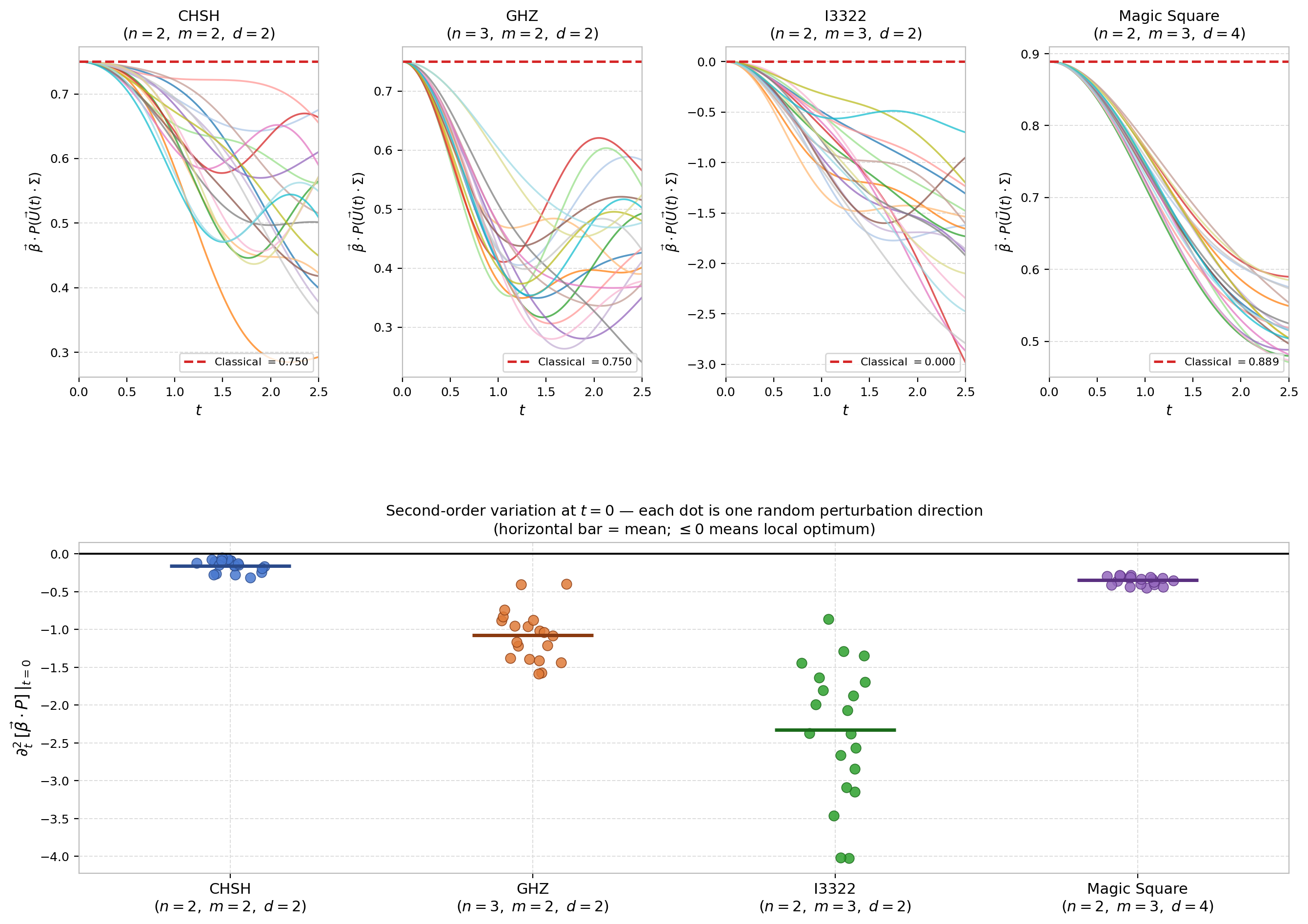}
     \caption{
        \textbf{Evolution of the evaluation of four Bell functionals under unitary perturbation of the optimal local deterministic strategy} (Lemma~\ref{lem:locopt}, see \ref{app:functionals} for more details)
        For each game, we draw $20$ trajectories by sampling a strategy $\Sigma (\rho, \Pi)$ realizing the best
        local deterministic correlation, (for all $i$, $(\Pi^{(i)}_{a_i \vert x_i})_{a_i \neq \tilde a_i(x_i)}$ is sampled from the Haar measure), and random perturbation generators $K_\rho$ and
        $(K_i^{(x)})_{i,x\geq 1}$, sampled with i.i.d $\mathcal{N}(0, 1)$ entries, skew-symmetrized and normalized.
        \textbf{(Top)} Evolution of the score $\vec{\beta}\cdot P(\vec{U}(t)\cdot\Sigma_0)$
        as a function of the perturbation parameter $t$; the dashed red line indicates
        the classical score.
        \textbf{(Bottom)} Second-order variation
        $\partial^2_t [\vec{\beta} \cdot P(\vec{U}(t)\cdot\Sigma_0)]\big|_{t=0}$
        estimated by finite differences for each trajectory (dots),
        with the mean indicated by a horizontal bar.
        All values are strictly negative for all four games, confirming
        Lemma~\ref{lem:locopt} for CHSH and GHZ ($(n,2,2)$ scenarios), and
        suggesting that it may extend to the $(n,m,d)$ setting, with $m > 2$. Interestingly enough, the result holds for the Magic Square game, even though $d > 2$. 
    }
    \label{fig:lemma7}
\end{figure}

\begin{proof}
    $p_{\tilde a}$ is in particular a strict one-flip maximum, so by Corollary \ref{cor:subsetneg}, it remains to show that the subset games are strictly negative on their active subspace. However, since $d = 2$ here,
    the subset games are scalars. More precisely, for all $S \subseteq \{1, \ldots, n\}$ with $|S| \ge 2$, $\Z_2^{x, S}$ consists of a single string
    which is the complementary of $\tilde a(x)_S$, and its lift $\alpha^{x, S}$ is equal to $\tilde a(x)$ outside of $S$. 
    \begin{equation}
        \bop_S(\Pi^S) = \sum_{\xi \in \Z_2^{|S|}} \sum_{x : x_S = \xi} (\beta_{\alpha^{x, S}, x} - \beta_{\tilde a(x), x}) = \sum_{x} \beta_{\alpha^{x, S}, x} - \beta_{\tilde a(x), x}
    \end{equation}
    Note that this is exactly the score increase, compared to $p_{\tilde a}$, of the local deterministic correlation realized by fixing the output of each player $i$ to 
    $\tilde a_i(x_i)$ if $i \in S$, and $\tilde a_i(x_i) \oplus 1$ if $i \notin S$. It is strictly negative since $p_{\tilde a}$ is a strict classical optimum.
\end{proof}

This is surprising : we know since that every extremal point of $\qset$ can be realized by a $\qp(2)$ correlation, so every Bell functional admits 
a maximum at a correlation in $\qp(2)$. Hence, if one could find extremal points of $Q$ arbitarily close de $p_{\tilde a}$ (which is the case if $\qset$ is curved around $p_{\tilde a})$, one could choose $\vec \beta$ 
such that $p_{\tilde a}$ is the strict optimal classical correlation, but the optimal quantum correlation is a quantum extremal point $p$ close from $p_{\tilde a}$. If $p$ is sufficiently close, one expects that,
slightly perturbating $p_{\tilde a}$ towards $p$ inside $\qp(2)$ achieves a larger score than $p_{\tilde a}$. This intuition is formalized in the proof of Theorem \ref{th:flat}.

\begin{proof}[Proof of Theorem \ref{th:flat}, see Fig. \ref{fig:flat}]
    We will need the following Theorem by Masanes \cite{masanes2005extremalquantumcorrelationsn}

    \begin{theorem}
        \label{thm:masanes}
        In the $(n, 2, 2)$ setting, $\ext(\qset) \subseteq \pqp(2)$.
    \end{theorem}

    Let $P$ be a two dimensional plane such that $p_{\tilde a} \in \qset^P \equiv \qset \cap P$. $\lset$ is a polytope having $p_{\tilde a}$ as a vertex, so $\lset^P \equiv \lset \cap P$ is a polygon and also has $p_{\tilde a}$ as a vertex, which is 
    adjacent to two edges $F_1^P = F_1 \cap P$ and $F_2^P = F_2 \cap P$ of $\lset^P$, where $F_1$ and $F_2$ are two faces of $\lset$. Let $H_1$ and $H_2$ be the half spaces delimited by $F_1$ and $F_2$ such that $H_1 \cap \lset = F_1$, and $H_2 \cap \lset = F_2$. Since $\lset^P \subseteq \qset^P$, it is clear that 
    $H_1 \cup H_2 \cap \partial \qset^P$ is a neighborhood of $p_{\tilde a}$ in $\partial \qset^P$ (where $\partial$ denotes the border). 
    We now treat $H_1 \cap \partial \qset^P$ and $H_2 \cap \partial \qset^P$ separately and show that if $p \in H_1 \cap \partial \qset^P$ is sufficiently close from $p_{\tilde a}$, then $p$ is not extremal (the reasoning for $H_2$ is the same).
    
    If $H_1 \cap \partial \qset^P = H_1 \cap \partial \lset^P$, then every point in $H_1 \cap \partial \qset^P$ lies on an edge and is hence not extremal (except for the extremities).
   
    Otherwise, we claim it is possible 
    to find a Bell functional $\vec \beta$ such that $p_{\tilde a}$ is a unique classical maximum for $\vec \beta$, and, if one defines $H_{\vec \beta} = \{p \in \R^{2^n \times d^n} \mid \vec \beta \cdot p \ge \vec \beta \cdot p_{\tilde a}\}$, and its associated separating hyperplane $F_{\vec \beta}$ then $(H_{\vec \beta} - F_{\vec \beta}) \cap \partial \qset^P  $ is still 
    a neighborhood of $p_{\tilde a}$ in $\partial \qset^P$. Let $\vec \beta_1$ be a Bell functional such that $F_1 = \{p \in \lset \mid \vec \beta_1 \cdot p = \vec \beta_1 \cdot p_{\tilde a} \}$. Since $p_{\tilde a}$ is a vertex of $\lset$, one can also find a Bell functional 
    $\vec \beta_{\tilde a}$ such that $p_{\tilde a}$ is a unique classical maximum for $\vec \beta_{\tilde a}$. Then, for all $\varepsilon > 0$, $p_{\tilde a}$ is a unique classical maximum for $\vec \beta_1 + \varepsilon \vec \beta_{\tilde a}$, defining a hyperplane $F_{\varepsilon} = \{p \in \R^{2^n \times d^n} \mid \vec \beta_\varepsilon \cdot p = \vec \beta_\varepsilon \cdot p_{\tilde a}\}$.
    The interior oriented angle $\theta(\varepsilon)$ between $F_1^P$ and the line $F_{\varepsilon} \cap P$ is then clearly continuous in $\varepsilon$, and always positive since otherwise some part of $\lset - \{p_{\tilde a}\}$ lies above $F_{\varepsilon}$ making $p_{\tilde a}$ suboptimal for $\vec \beta_{\varepsilon}$. We also have $\theta(0) = 0$ by defintiion. Moreover,
    since $H_1 \cap \partial \lset^P \neq H_1 \cap \partial \qset^P$, the interior oriented angle $\theta$ between $F_1^P$ and the convex curve $H_1 \cap \partial \qset^P$ must be strictly positive. This means there must exist $\varepsilon_0$ such that $0 < \theta(\varepsilon_0)< \theta$, in which case $\vec \beta_{\varepsilon_0}$ satisfies the desired property.
    By Lemma \ref{lem:locopt}, $p_{\tilde a}$ is a local optimum in $\qp$ for $\vec \beta_{\varepsilon_0}$. This is equivalent to saying that there exists $r_{P, 1}$ such that $\qp(2) \cap (H_{\varepsilon_0} - F_{\varepsilon_0}) \cap B(p_{\tilde a}, r_{P, 1}) = \emptyset$. For $r_{P, 1}$ small enough, since $\partial \qset^P \cap (H_{\varepsilon_0} - F_{\varepsilon_0})$ is a neighborhood 
    of $p_{\tilde a}$ in $\partial \qset^P$, this implies $\qp \cap \partial \qset^P \cap B(p_{\tilde a}, r_{P, 1}) = \emptyset$. Theorem \ref{thm:masanes} concludes.
\end{proof}

\begin{figure}
\tikzset{every picture/.style={line width=0.75pt}} 

\begin{tikzpicture}[x=1pt,y=1pt,yscale=-1,xscale=1]

\draw [color={rgb, 255:red, 189; green, 16; blue, 224 }  ,draw opacity=1 ]   (34.7,23.4) -- (284.77,143.04) ;
\draw  [draw opacity=0][fill={rgb, 255:red, 74; green, 144; blue, 226 }  ,fill opacity=0.22 ] (83.99,-0.2) -- (265.59,219.8) -- (305.59,220.2) -- (305.59,0.6) -- cycle ;
\draw  [color={rgb, 255:red, 208; green, 2; blue, 27 }  ,draw opacity=1 ][fill={rgb, 255:red, 255; green, 0; blue, 0 }  ,fill opacity=1 ] (145.89,78.64) .. controls (145.89,76.8) and (147.38,75.3) .. (149.23,75.3) .. controls (151.08,75.3) and (152.57,76.8) .. (152.57,78.64) .. controls (152.57,80.49) and (151.08,81.99) .. (149.23,81.99) .. controls (147.38,81.99) and (145.89,80.49) .. (145.89,78.64) -- cycle ;
\draw [color={rgb, 255:red, 126; green, 211; blue, 33 }  ,draw opacity=1 ]   (152.57,78.64) -- (214.44,94.53) ;
\draw    (81.51,96.2) -- (145.89,78.64) ;
\draw    (14.62,160.58) .. controls (22.14,130.48) and (48.06,106.24) .. (81.51,96.2) ;
\draw [color={rgb, 255:red, 126; green, 211; blue, 33 }  ,draw opacity=1 ]   (281.75,161.42) .. controls (275.06,126.3) and (248.31,104.56) .. (214.44,94.53) ;
\draw    (147.2,80.81) -- (40.85,201.73) ;
\draw [color={rgb, 255:red, 74; green, 144; blue, 226 }  ,draw opacity=1 ]   (151.2,81.21) -- (248.31,198.21) ;
\draw [color={rgb, 255:red, 255; green, 0; blue, 0 }  ,draw opacity=1 ]   (149.23,78.64) -- (149.4,51.2) ;
\draw [shift={(149.42,48.2)}, rotate = 90.36] [fill={rgb, 255:red, 255; green, 0; blue, 0 }  ,fill opacity=1 ][line width=0.08]  [draw opacity=0] (3.57,-1.72) -- (0,0) -- (3.57,1.72) -- cycle    ;
\draw [color={rgb, 255:red, 74; green, 144; blue, 226 }  ,draw opacity=1 ]   (149.23,78.64) -- (166.38,63.39) ;
\draw [shift={(168.62,61.4)}, rotate = 138.35] [fill={rgb, 255:red, 74; green, 144; blue, 226 }  ,fill opacity=1 ][line width=0.08]  [draw opacity=0] (3.57,-1.72) -- (0,0) -- (3.57,1.72) -- cycle    ;
\draw [color={rgb, 255:red, 189; green, 16; blue, 224 }  ,draw opacity=1 ]   (149.23,78.64) -- (158.81,54.19) ;
\draw [shift={(159.9,51.4)}, rotate = 111.39] [fill={rgb, 255:red, 189; green, 16; blue, 224 }  ,fill opacity=1 ][line width=0.08]  [draw opacity=0] (3.57,-1.72) -- (0,0) -- (3.57,1.72) -- cycle    ;
\draw  [draw opacity=0] (176.61,90.91) .. controls (176.22,91.78) and (175.78,92.64) .. (175.3,93.49) .. controls (173.58,96.51) and (171.41,99.12) .. (168.93,101.28) -- (149.23,78.64) -- cycle ; \draw  [color={rgb, 255:red, 189; green, 16; blue, 224 }  ,draw opacity=1 ] (176.61,90.91) .. controls (176.22,91.78) and (175.78,92.64) .. (175.3,93.49) .. controls (173.58,96.51) and (171.41,99.12) .. (168.93,101.28) ;  
\draw  [draw opacity=0] (159.21,80.95) .. controls (158.76,83.1) and (157.63,85.09) .. (155.92,86.56) -- (149.23,78.64) -- cycle ; \draw  [color={rgb, 255:red, 126; green, 211; blue, 33 }  ,draw opacity=1 ] (159.21,80.95) .. controls (158.76,83.1) and (157.63,85.09) .. (155.92,86.56) ;  

\draw (9.9,182.62) node [anchor=north west][inner sep=0.75pt]    {$P$};
\draw (50,125) node [anchor=north west][inner sep=0.75pt][font = \LARGE]    {$\mathcal{Q}^{P}$};
\draw (138,154.91) node [anchor=north west][inner sep=0.75pt][font = \LARGE]  {$\mathcal{L}^{P}$};
\draw (158,117) node [anchor=north west][inner sep=0.75pt]  [font = \Large, color={rgb, 255:red, 74; green, 144; blue, 226 }  ,opacity=1 ]  {$F_{1}^{P}$};
\draw (119.6,115.4) node [anchor=north west][inner sep=0.75pt][font = \Large] {$F_{2}^{P}$};
\draw (144.93,82.82) node [anchor=north west][inner sep=0.75pt]  [font=\footnotesize,color={rgb, 255:red, 255; green, 0; blue, 0 }  ,opacity=1 ]  {$p_{\tilde{a}}$};
\draw (165.2,28.6) node [anchor=north west][inner sep=0.75pt]  [font= \Large, color={rgb, 255:red, 74; green, 144; blue, 226 }  ,opacity=1 ]  {$H_{1}$};
\draw (171.2,55.2) node [anchor=north west][inner sep=0.75pt]  [font=\footnotesize,color={rgb, 255:red, 74; green, 144; blue, 226 }  ,opacity=1 ]  {$\vec{\beta }_{1}$};
\draw (138,41.49) node [anchor=north west][inner sep=0.75pt]  [font=\footnotesize,color={rgb, 255:red, 255; green, 0; blue, 0 }  ,opacity=1 ]  {$\vec{\beta }_{\tilde{a}}$};
\draw (153.2,39.2) node [anchor=north west][inner sep=0.75pt]  [font=\footnotesize,color={rgb, 255:red, 189; green, 16; blue, 224 }  ,opacity=1 ]  {$\overrightarrow{\beta _{\varepsilon }}$};
\draw (89.2,37) node [anchor=north west][inner sep=0.75pt]  [font= \Large, color={rgb, 255:red, 189; green, 16; blue, 224 }  ,opacity=1 ]  {$F_{\varepsilon }$};
\draw (172.8,98.71) node [anchor=north west][inner sep=0.75pt]  [font=\footnotesize, color={rgb, 255:red, 189; green, 16; blue, 224 }  ,opacity=1 ]  {$\theta ( \varepsilon )$};
\draw (158.01,73.55) node [anchor=north west][inner sep=0.75pt]  [font=\footnotesize,color={rgb, 255:red, 126; green, 211; blue, 33 }  ,opacity=1 ]  {$\theta $};
\draw (208.8,78.8) node [anchor=north west][inner sep=0.75pt]  [font=\Large ,color={rgb, 255:red, 126; green, 211; blue, 33 }  ,opacity=1 ]  {$\partial \mathcal{Q}^{P} \cap H_{1}$};

\end{tikzpicture}

\caption{Illustration for the proof of Theorem \ref{th:flat}. Here, $p_{\tilde a}$ is the unique classical optimum for $\beta_\varepsilon$, and every point above $F_{\varepsilon}$ in $\partial \qset^P$ is better than 
$p_{\tilde a}$ for $\vec \beta_{\varepsilon}$. Thus, points too close from $p_{\tilde a}$ in $\partial \qset^P \cap H_1$ cannot be in $\qp(2)$, hence are not extremal.}
\label{fig:flat}
\end{figure}
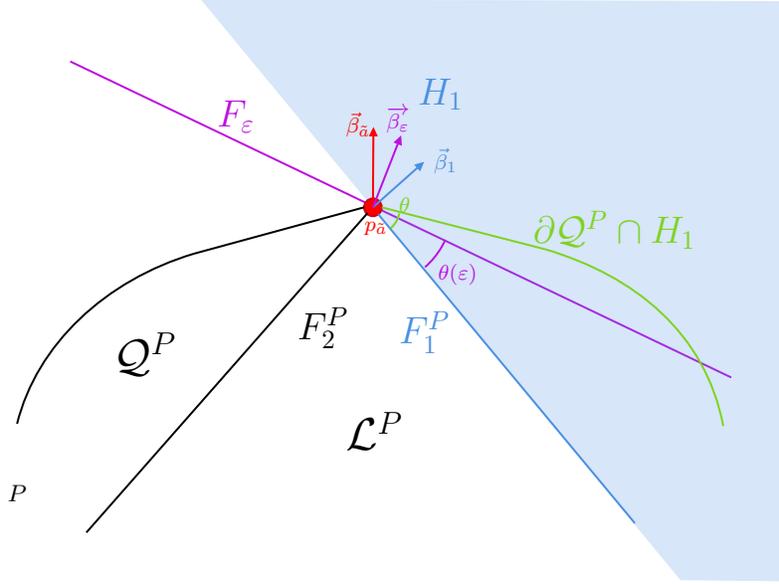

To conclude this section, we insist on the formulation of Theorem \ref{th:flat}: restricting to a two-dimensional plane is essential, 
and one cannot find a radius $r > 0$ such that $\ext(\qset) \cap B(p_{\tilde a}, r) = \{p_{\tilde a}\}$. Indeed, \cite{barizien2025quantum} gives 
necessary and sufficient conditions for a strategy to realize an extremal point of $\qset$ in the $(2, 2, 2)$ scenario. These allow
to find a sequence of strategies converging to a local deterministic correlation, while remaining extremal. While this may seem contradictory, 
simple shapes have this property of local flatness in every plane around a point $p$ with extremal points arbitarily close to $p$, even in $3$ dimensions (see Appendix \ref{app:shapes}).

%% file: sections_journal_v3/conclusion.tex
\section{Discussion and Open questions} 
By analyzing the dynamics of quantum correlations under unitary perturbations of strategies, we have uncovered new properties of the quantum correlation set $\mathcal{Q}$ as well as provided a powerful tool for analyzing local optimality via the reduction of multipartite Bell optimization to subset games.
Of direct practical interest, the better understanding of the geometry of \(\qset\) highlights that the efficacy of variational methods for learning optimal strategies generally depends on the dimensionality of the ansatz state. Our proof of the flatness of $\mathcal{Q}$ around deterministic points in the $(n, 2, 2)$ scenario is a first step towards understanding the geometry of the multipartite quantum set. Furthermore, the distinction we draw between projective and general POVM perturbations offers a concrete roadmap for resolving Gisin's problem.

\textbf{Ansatz Dimension as a Learning Resource.}
Our results underscore the role of the \textit{ansatz dimension} in variational quantum algorithms. Indeed, we show that classical optimum correlations are always locally optimal in the set of qubit projective strategies. This means that even in the $(n, 2, 2)$ case where the optimal violation is known to be achievable with qubits ($D=2$), a variational quantum circuit --- using parametrized shared state preparation and local measurements --- with local dimension $2$ initialized near a classical point may fail to escape the local optimum. The convexity of $\mathcal{Q}$ and of Bell functionals suggests that increasing the local dimension of the ansatz can eliminate these local optima. This positions the dimension of the variational circuit as a resource for learning, independently on the target dimension.

\textbf{New games exhibiting quantum--classical gap.} The local optimality criterion via subset games (Corollary \ref{cor:subsetneg}) suggests a way to construct new games with a quantum--classical gap. It suffices to consider a $(k, 2, d-1)$ game $B$
with strictly positive value, and construct a game for which one of the subset games at the optimal classical correlation $p_{\tilde a}$ is $B$. Then, $p_{\tilde a}$ is not a local optimum for $B$ in $\mathcal{Q}_P(d)$,
and a fortiori the game has a quantum gap. It would be interesting to investigate if this construction can succeed and yield any nontrivial new type of game with a quantum--classical gap.

\textbf{POVMs vs. PVMs}
This work sheds new light on the open question posed by Gisin: does there exist a correlation achievable by POVMs of dimension $D$ that is not in $\mathcal{Q}_P(D)$? 
Our analysis reveals a structural difference between PVMs and POVMs under perturbation. For PVMs, the measurement evolution terms in the second-order expansion vanish or are strictly negative due to orthogonality constraints. For general POVMs, non-orthogonal elements allow these terms to potentially contribute positively. We suggest that constructing a Bell functional where the subset games are negative but POVM perturbations can yield a positive gain would prove $\mathcal{Q}_P(D) \subsetneq \mathcal{Q}(D)$. Our perturbative machinery provides the explicit conditions required to search for such a functional.

\section*{Acknowledgements}
The authors thank Ulysse Chabaud and Ivan Šupić for the discussions that lead to this research project. S.C. thanks Jean-Daniel Bancal for his insights on the extremal points of $\qset$ in the $(2, 2, 2)$ scenario, and Marc-Olivier Renou and the PhiQus team for introducing him to the POVM vs. PVM conjecture. This work received partial support from France 2030 under the French National Research Agency award number ANR-22-PNCQ-0002 (HQI).

%% file: sections_journal_v3/appendix.tex
\section{Details of the functionals used in Figure \ref{fig:lemma7}}
\label{app:functionals}
All four games are cast as Bell functionals $\vec{\beta}$ on the correlations $p(a\vert x)$, with inputs $x_i \in \mathbb{Z}_m$
and outputs $a_i \in \mathbb{Z}_d$.

\medskip
\noindent\textbf{CHSH} $(n=2,\, m=2,\, d=2)$.
The winning condition is $a_1 \oplus a_2 = x_1 \wedge x_2$, with uniform input
distribution. The functional is:
\begin{equation}
    \beta_{(a_1,a_2),(x_1,x_2)} = \tfrac{1}{4}\,
    \mathbf{1}\bigl[a_1 \oplus a_2 = x_1 x_2\bigr].
\end{equation}
The best local deterministic strategy is $\tilde{a}_i \equiv 0$, achieving
$\vec{\beta} \cdot p_{\tilde{a}} = 3/4$. The quantum maximum is
$(2+\sqrt{2})/4$ (Tsirelson bound~\cite{tsirelson1980quantum}).

\medskip
\noindent\textbf{GHZ / Mermin} $(n=3,\, m=2,\, d=2)$.
Inputs are constrained to $x_1+x_2+x_3 \equiv 0 \pmod{2}$,
i.e.\ $x \in \{000, 110, 101, 011\}$, drawn uniformly. The winning condition is:
\begin{equation}
    a_1 \oplus a_2 \oplus a_3 =
    \begin{cases} 0 & \text{if } x = 000, \\ 1 & \text{otherwise,}\end{cases}
\end{equation}
giving $\beta_{a,x} = \frac{1}{4}\mathbf{1}[x \in \{000,110,101,011\}]
\cdot \mathbf{1}[a_1\oplus a_2\oplus a_3 = \lfloor(x_1+x_2+x_3)/2\rfloor]$.
The classical maximum $3/4$ is achieved by e.g.\ $\tilde{a} = (0,0,1)$; the quantum value is $1$~\cite{greenberger1990bell}.

\medskip
\noindent $\textbf{I}_{3322}$ $(n=2,\, m=3,\, d=2)$.
The $I_{3322}$ functional~\cite{Collins_2004} reads:
\begin{align}
    I_{3322} &= \sum_{(x_1,x_2) \neq (2,2)} p(a_1{=}a_2 \mid x_1 x_2)
               - p(a_1{=}a_2 \mid 22) \nonumber \\
             &\quad - 2p_1(0\mid 0) - p_1(0\mid 1)
               - 2p_2(0\mid 0) - p_2(0\mid 1),
    \label{eq:i3322_compact}
\end{align}
where $p(a_1{=}a_2\mid x_1x_2) = p(00\mid x_1x_2) + p(11\mid x_1x_2)$
and $p_i(0\mid x_i)$ denotes the marginal of party $i$.
Local correlations satisfy $I_{3322} \leq 0$~\cite{Collins_2004},
with the classical maximum $0$ achieved by the constant strategy
$\tilde{a}_i \equiv 1$.

The corresponding Bell functional vector $\vec{\beta}$, with coordinates
$\beta_{(a_1,a_2),(x_1,x_2)}$ indexed by outcomes $(a_1,a_2)\in\mathbb{Z}_2^2$
and inputs $(x_1,x_2)\in\mathbb{Z}_3^2$, is obtained by expanding the
marginals as $p_i(0\mid x_i) = \frac{1}{3}\sum_{x_j=0}^{2}\sum_{a_j}
p(0,a_j\mid x_i,x_j)$, giving:
\begin{equation}
    \beta_{(a_1,a_2),(x_1,x_2)} =
    s(x_1,x_2)\,\mathbf{1}[a_1=a_2]
    - \frac{c_1(x_1)}{3}\,\mathbf{1}[a_1=0]
    - \frac{c_2(x_2)}{3}\,\mathbf{1}[a_2=0],
    \label{eq:i3322_beta}
\end{equation}
where
\begin{equation}
    s(x_1,x_2) =
    \begin{cases}
        +1 & \text{if } (x_1,x_2) \in
            \{(0,0),(1,0),(0,1),(1,1),(2,1),(0,2),(1,2)\}, \\
        -1 & \text{if } (x_1,x_2) = (2,2), \\
         0 & \text{otherwise},
    \end{cases}
\end{equation}
and $c_i(x_i)$ encodes the marginal coefficients:
\begin{equation}
    c_1(x_1) =
    \begin{cases}
        2 & \text{if } x_1 = 0, \\
        1 & \text{if } x_1 = 1, \\
        0 & \text{if } x_1 = 2,
    \end{cases}
    \qquad
    c_2(x_2) =
    \begin{cases}
        2 & \text{if } x_2 = 0, \\
        1 & \text{if } x_2 = 1, \\
        0 & \text{if } x_2 = 2.
    \end{cases}
\end{equation}

\medskip
\noindent\textbf{Magic Square} $(n=2,\, m=3,\, d=4)$.
Party 1 receives a row index $r \in \{0,1,2\}$ and outputs one of the four
even-parity binary triples
$\mathcal{T}_1 = \{000, 011, 101, 110\}$, encoded as $a_1 \in \{0,1,2,3\}$.
Party 2 receives a column index $c \in \{0,1,2\}$ and outputs one of the four
odd-parity binary triples $\mathcal{T}_2 = \{001, 010, 100, 111\}$,
encoded as $a_2 \in \{0,1,2,3\}$.
The winning condition is that the two parties agree on the shared cell:
$\mathcal{T}_1(a_1)[c] = \mathcal{T}_2(a_2)[r]$.
With uniform input distribution over $\{0,1,2\}^2$:
\begin{equation}
    \beta_{(a_1,a_2),(r,c)} =
    \tfrac{1}{9}\,\mathbf{1}\bigl[\mathcal{T}_1(a_1)[c] = \mathcal{T}_2(a_2)[r]\bigr].
\end{equation}
The classical maximum $8/9$ is achieved by the input-dependent strategy
$\tilde{a}_1 = (0,0,1)$, $\tilde{a}_2 = (0,0,0)$ (in triple-index encoding),
while the quantum value is $1$~\cite{aravind2003}.
The quantum strategy realizing $p_{\tilde{a}}$ is constructed following
Lemma~\ref{lem:localdeterministic}: the shared state is $\rho_0 = |0\rangle\langle 0|^{\otimes 2}$
in $\mathbb{C}^4 \otimes \mathbb{C}^4$, and for each party $i$ and input $x_i$,
the winning projector $\Pi^{(i)}_{\tilde{a}_i(x_i)|x_i} = |0\rangle\langle 0|$
is the same for all $x_i$, while the remaining three projectors are drawn
Haar-randomly in $\mathrm{span}(e_1, e_2, e_3)$, identically across inputs.

\section{Additional discussions on the geometry of $\qset$ around local deterministic points}
\label{app:shapes}

Here, we explain why one cannot find a ball around $p_{\tilde a}$ with no extremal point, despite Theorem $\ref{th:flat}$ stating it is true in any $2$-dimensional plane containing $p_{\tilde a}$.

In \cite{barizien2025quantum}, Barizien and Bancal gave necessary and sufficient conditions for a correlation to be extremal in the $(2,2,2)$ scenario. Up to local isometries and relabelings, it is sufficient to consider realizations consisting of a pure, partially entangled two-qubit state and real projective measurements in the $XZ$ plane.

The state is parameterized by an entanglement angle $\theta \in [0, \pi)$:
\begin{equation}
    |\phi_\theta\rangle = \cos(\theta)|00\rangle + \sin(\theta)|11\rangle.
\end{equation}
The measurements for Alice ($A_x$) and Bob ($B_y$), for inputs $x,y \in \{0,1\}$, are parameterized by angles $a_x$ and $b_y$:
\begin{align}
    A_x &= \cos(a_x)\sigma_z + \sin(a_x)\sigma_x, \\
    B_y &= \cos(b_y)\sigma_z + \sin(b_y)\sigma_x.
\end{align}
Using the symmetries of the set, these parameters can be restricted to the range $0 \le a_0 \le b_0 \le b_1 < \pi$ and $a_0 \le a_1 < \pi$.

This specific characterization relies on a non-linear steering map that defines modified measurement angles and correlators. For an outcome $a \in \{\pm 1\}$, the modified measurement angles for Alice's operators are defined as:
\begin{equation}
    \tilde{a}_x^a = 2 \arctan\left(\tan\left(\frac{a_x}{2}\right)\tan(\theta)^a\right).
\end{equation}
Similarly, the corresponding modified correlators are constructed from the original statistics as:
\begin{equation}
    [\tilde{A}_x^a B_y] = \frac{\langle A_x B_y \rangle + a \langle B_y \rangle}{1 + a \langle A_x \rangle}.
\end{equation}
With these definitions in place, we can state the full characterization of the extremal points.

\begin{theorem}[Characterization of $\text{Ext}(\mathcal{Q})$, Barizien \& Bancal]
In the minimal CHSH scenario, the following hold:
\begin{enumerate}
    \item A nonlocal point is extremal in the quantum set $\mathcal{Q}$ if and only if for all $u \in \{\pm 1\}^2$,
    \begin{equation}
        \sum_{x,y} \epsilon_{xy} \arcsin[\tilde{A}_x^{u_x}B_y] = \pi
    \end{equation}
    for some $\epsilon_{xy} \in \{\pm 1\}$ such that $\prod_{x,y} \epsilon_{xy} = -1$.
    \item A quantum realization leads to a nonlocal extremal point if and only if it can be mapped by local channels and relabelings to a quantum realization on the entangled state $|\phi_\theta\rangle$ with measurements satisfying the full alternation condition:
    \begin{equation}
        \forall (s,t) \in \{\pm 1\}^2, \quad 0 \le [\tilde{a}_0^s]_\pi \le b_0 \le [\tilde{a}_1^t]_\pi \le b_1 < \pi,
    \end{equation}
    where $[\alpha]_\pi \equiv \alpha \pmod{\pi}$.
\end{enumerate}
\end{theorem}

Condition $2$ is particularly interesting for our purpose. Note that if $a_0 = a_1 = b_0 = b_1 = 0$, then the strategy realizes the local deterministic correlation $p_0$ (always output $0$).
Thus, making the measurement angles tend to $0$ while keeping the full alternation condition yields extremal correlations arbitrarily close to $p_0$. The following choice of angles, depending on 
the parameter of entanglement $\theta$, is due to Barizien and Bancal.

\begin{proposition}[Sequence of extremal correlations tending to $p_0$]
    The following measurement angles realize extremal correlations tending to $p_0$ when $\theta \rightarrow 0$
    \begin{align}
        a_0 = 0, \quad a_1 = \sqrt \theta, \quad b_0 = \tilde a_1^+,\quad b_1 = \tilde a_1^-.
    \end{align}
\end{proposition}

While this may come as a surprise in the light of Theorem $\ref{th:flat}$, note that this phenomenon can occur even in dimension $3$. 

\begin{proposition}
    Consider the sequence of $3$-dimensional points $p_n = (1/n, 1/n^2, 1/n^3)$, defined for $n \ge 1$, and
    $C = \conv \{(0, 0, 0), (p_n)_{n \ge 1} \}$. Then, $p_n$ is a sequence of extremal points converging to $(0, 0, 0)$, and $\partial C \cap P$ 
    contains no extremal point of $C$ in a sufficiently small neighborhood of $(0, 0, 0)$.
\end{proposition}

\begin{proof}
    By definition, $\ext(C) \subseteq \{ (0, 0, 0) \} \cup \{p_n, n \ge 1\}$. Let us show the reverse inclusion. The point $(0, 0, 0)$ is clearly extremal, so this amounts to showing that none of the points $p_n$
    can be expressed as a convex combination of a finite number of the points $(0, 0, 0)$ and $p_n$, or equivalently that none of the points $p_n$ lies in the cone generated by a 
    finite number of the other points $p_n$. Suppose there exists $r, n_1, \ldots, n_r, n \in \mathbb{N}^*$, and $\lambda_1, \ldots, \lambda_r > 0$ such that:
    \begin{equation}
        \frac{1}{n} = \sum_{i=1}^r \frac{\lambda_i }{n_i}, \qquad \frac{1}{n^2} = \sum_{i=1}^r \frac{\lambda_i }{n_i^2}, \qquad \frac{1}{n^3} = \sum_{i=1}^r \frac{\lambda_i }{n_i^3}.
    \end{equation}

    Set $\mu_i = \frac{\lambda_i n^2}{n_i^2}$, so that $\sum_i \mu_i = 1$ and $\frac{1}{n^2} = \sum_{i = 1}^r \mu_i \frac{1}{n_i^2}$.
    Then, $\frac{1}{n}^3 = \sum_{i = 1}^r \frac{\lambda_i}{n_i^3} = \frac{1}{n^2} \sum_{i = 1}^r \mu_i \frac{1}{n_i}$, so $\frac{1}{n} = \sum_{i = 1}^r \mu_i \frac{1}{n_i}$.

    This yields the equality $\frac{1}{n^2} = \left(\sum_{i = 1}^r \mu_i \frac{1}{n_i}\right)^2 = \sum_{i = 1}^r \mu_i \frac{1}{n_i^2}$. By strict convexity of $t \mapsto t^2$, this imposes $\frac{1}{n_1} = \cdots = \frac{1}{n_r}$, and thus $n = n_1$.
    
    Moreover, if $P$ is a plane, then there exists $a, b, c$ such that $p_n \in P \Leftrightarrow \frac{a}{n} + \frac{b}{n^2} + \frac{c}{n^3} = 0$.
    This polynomial equation of degree $3$ in $\frac{1}{n}$ has at most $3$ solutions, hence at most three points $p_n$ belong to $P$. Any point in $C \cap P$ closer to $(0, 0, 0)$ than all of these three points is thus non-extremal.
\end{proof}